\documentclass[a4paper,UKenglish,cleveref, autoref,thm-restate]{lipics-v2021}

\usepackage[ruled,vlined]{algorithm2e}

\usepackage{booktabs}
\usepackage[framemethod=tikz,nobreak=true]{mdframed}

\DeclareMathOperator{\ecc}{ecc}
\DeclareMathOperator*{\E}{\mathbb{E}}
\DeclareMathOperator{\cut}{cut}
\DeclareMathOperator{\GKP}{GKP}
\DeclareMathOperator{\ex}{ex}
\DeclareMathOperator{\poly}{poly}
\DeclareMathOperator{\polylog}{polylog}
\DeclareMathOperator{\define}{\overset{def}{=}}
\DeclareMathOperator{\disj}{\mathsf{DISJ}}
\DeclareMathOperator{\congclique}{\mathsf{CLIQUE}}
\DeclareMathOperator{\congest}{\mathsf{CONGEST}}
\DeclareMathOperator{\ncc}{\mathsf{NCC}}
\DeclareMathOperator{\local}{\mathsf{LOCAL}}
\DeclareMathOperator{\oracle}{\mathsf{Oracle}}

\DeclareMathOperator{\hybrid}{\mathsf{HYBRID}}
\DeclareMathOperator{\bcc}{\mathsf{BCC}}

\title{Deterministic Distributed Algorithms and Lower Bounds in the Hybrid Model} 

\titlerunning{Deterministic Distributed Algorithms and Lower Bounds in the Hybrid Model} 

\author{Ioannis Anagnostides}{Department of Computer Engineering, National Technical University of Athens, Greece}{ioannis.anagnostides@gmail.com}{}{}

\author{Themis Gouleakis}{Max Planck Institute for Informatics, Saarbrucken, Germany}{tgouleak@mpi-inf.mpg.de}{}{}

\authorrunning{I. Anagnostides and T. Gouleakis} 

\Copyright{Ioannis Anagnostides and Themis Gouleakis} 



\category{} 

\relatedversion{} 

\acknowledgements{We are indebted to Christoph Lenzen for carefully reviewing an earlier draft of our work, and proposing several improvements and interesting directions. Specifically, he suggested derandomizing \Cref{proposition:sparse} via the Garay-Kutten-Peleg algorithm, while he also pointed out the connection with low-congestion shortcuts, leading to the results of \Cref{section:simulation}. We are also very grateful to the anonymous reviewers at DISC for carefully reviewing this paper, and for indicating many corrections and ways to improve the exposition. We are particularly thankful to a reviewer for providing very detailed arguments which strengthened our results in \Cref{subsubsubsection:apsp}. All errors remain our own.}

\nolinenumbers 

\hideLIPIcs  

\EventEditors{John Q. Open and Joan R. Access}
\EventNoEds{2}
\EventLongTitle{42nd Conference on Very Important Topics (CVIT 2016)}
\EventShortTitle{CVIT 2016}
\EventAcronym{CVIT}
\EventYear{2016}
\EventDate{December 24--27, 2016}
\EventLocation{Little Whinging, United Kingdom}
\EventLogo{}
\SeriesVolume{42}
\ArticleNo{23}

\begin{document}

\maketitle

\begin{abstract}
The $\hybrid$ model was recently introduced by Augustine et al. \cite{DBLP:conf/soda/AugustineHKSS20} in order to characterize from an algorithmic standpoint the capabilities of networks which combine multiple communication modes. Concretely, it is assumed that the standard $\local$ model of distributed computing is enhanced with the feature of all-to-all communication, but with very limited bandwidth, captured by the node-capacitated clique ($\ncc$). In this work we provide several new insights on the power of hybrid networks for fundamental problems in distributed algorithms.

First, we present a deterministic algorithm which solves any problem on a sparse $n$-node graph in $\widetilde{\mathcal{O}}(\sqrt{n})$ rounds of $\hybrid$, where the notation $\widetilde{\mathcal{O}}(\cdot)$ suppresses polylogarithmic factors of $n$. We combine this primitive with several sparsification techniques to obtain efficient distributed algorithms for general graphs. Most notably, for the all-pairs shortest paths problem we give deterministic $(1 + \epsilon)$- and $\log n/\log \log n$-approximate algorithms for unweighted and weighted graphs respectively with round complexity $\widetilde{\mathcal{O}}(\sqrt{n})$ in $\hybrid$, closely matching the performance of the state of the art randomized algorithm of Kuhn and Schneider \cite{10.1145/3382734.3405719}. Moreover, we\footnote{See the acknowledgments.} make a connection with the Ghaffari-Haeupler framework of low-congestion shortcuts \cite{DBLP:conf/soda/GhaffariH16}, leading---among others---to a $(1 + \epsilon)$-approximate algorithm for Min-Cut after $\mathcal{O}(\polylog (n))$ rounds, with high probability, even if we restrict local edges to transfer $\mathcal{O}(\log n)$-bits per round. Finally, we prove via a reduction from the set disjointness problem that $\widetilde{\Omega}(n^{1/3})$ rounds are required to determine the radius of an unweighted graph, as well as a $(3/2 - \epsilon)$-approximation for weighted graphs. As a byproduct, we show an $\widetilde{\Omega}(n)$ round-complexity lower bound for computing a $(4/3 - \epsilon)$-approximation of the radius in the broadcast variant of the congested clique, even for unweighted graphs.
\end{abstract}

\section{Introduction}

Hybrid networks have found numerous applications in real-life computer systems. Indeed, leveraging different communication modes has substantially reduced the complexity and has improved the efficiency of the system, measured in terms of the energy consumption, the latency, the number of switching links, etc. For instance, hybrid architectures have been extensively employed in data centers, augmenting the traditional electrical switching architecture with optical switches in order to establish direct connections \cite{CHEN201645,10.1145/1851182.1851223,10.1145/1851182.1851222}. Another notable example is the 5G standard, which enhances the traditional cellular infrastructure with device-to-device (D2D) connections in order to guarantee very low latency among communication users (see \cite{6736752,KAR2018203,7147775,7588891}, and references therein).

Despite the central role of hybrid architectures in communication systems, a rigorous investigation of their potential has only recently began to formulate in the realm of distributed algorithms. In particular, Augustine et al. \cite{DBLP:conf/soda/AugustineHKSS20} proposed $\hybrid$, a model which combines the extensively-studied \emph{local} ($\local$) \cite{10.1137/0221015,10.5555/355459} model with the recently introduced \emph{node-capacitated clique} ($\ncc$) \cite{10.1145/3323165.3323195}. The former model captures the \emph{locality} of a given problem---nodes are able to exchange messages of arbitrary size but only with adjacent nodes, while the latter model---which enables all-to-all communication but with severe capacity restrictions for every node---addresses the issue of \emph{congestion}; these constitute the main challenges in distributed computing. From a practical standpoint, the local network captures the capabilities of \emph{physical} networks, wherein dedicated edges (e.g. cables or optical fibers) offer large bandwidth and high efficiency, but lack flexibility as they cannot be dynamically adapted by the nodes. In contrast, the global mode relates to \emph{logical} networks, which are formed as an overlay over a shared physical network \cite{feldmann2020fast}; here the feature of all-to-all communication comes at the cost of providing very limited throughput.

In this work we follow the recent line of research \cite{DBLP:conf/soda/AugustineHKSS20,10.1145/3382734.3405719,feldmann2020fast,censorhillel2020distance,gotte2020timeoptimal,DBLP:conf/spaa/Censor-HillelLP21} which endeavors to explore from a theoretical standpoint the power of hybrid models in distributed computing; specifically, the main issue that arises is whether combining two different communication modes offers a substantial improvement over each mode separately. This question is answered in the affirmative for a series of fundamental problems in distributed algorithms, while we also provide some hardness results mainly based on well-established communication-complexity lower bounds.

\subsection{Contributions \& Techniques}

\subsubsection{Sparsification}

First, we consider the design of $\hybrid$ algorithms in \emph{sparse} graphs---i.e. the average degree is $\polylog n$, where $n$ represents the number of communication entities in the graph. We prove the following general result:

\begin{restatable}{theorem}{deterministic}
\label{theorem:deterministic_sparse}
Consider a graph $G = (V, E, w)$ with $|E| = \widetilde{\mathcal{O}}(n)$. There exists a deterministic distributed algorithm such that every node learns the entire topology of the graph in $\widetilde{\mathcal{O}}(\sqrt{n})$ rounds of $\hybrid$.
\end{restatable}

As a warm-up, we first provide a \emph{randomized} $\hybrid$ algorithm so that every node learns the topology in $\widetilde{\mathcal{O}}(\sqrt{n})$ rounds. More precisely, observe that it suffices to solve a specific instance of the \emph{gossip problem} wherein every node $u \in V$ has to broadcast $\deg(u)$ number of $\mathcal{O}(\log n)$-bit messages---corresponding to its adjacency list---to all the other nodes in the graph. In this context, directly executing the \emph{token dissemination} protocol of Augustine et al. \cite{DBLP:conf/soda/AugustineHKSS20} requires in the worst-case a linear number of rounds since the high-degree nodes create a substantial communication bottleneck. However, we observe that there is a simple remedy; namely, the nodes with high degree can perform \emph{load balancing} via their local neighborhood. Interestingly, this idea relates to the \emph{density-aware} model proposed by Censor-Hillel et al. \cite{censorhillel2020distance}, which they refer to as the $\oracle$ model, in which the broadcasting capacity of a node depends on its degree.

More importantly, we also present a deterministic communication pattern for sparse networks (\Cref{theorem:deterministic_sparse}). Specifically, we first employ the Garay-Kutten-Peleg algorithm \cite{DBLP:journals/siamcomp/GarayKP98} in order to construct a "balanced" partition of the nodes, so that every cluster has "small" weak diameter. Then, we present several deterministic subroutines which allow to disseminate the composition of the clusters, perform load balancing, and finally broadcast the topology to the entire network within the desired round complexity. Along the way, we derandomize the token dissemination protocol of Augustine et al. \cite{DBLP:conf/soda/AugustineHKSS20}, which is one of their main communication primitives. 

Naturally, our guarantee for sparse graphs has an independent interest given that most communication networks of practical interest are very sparse \cite{snapnets}; the canonical example typically cited is the Internet \cite{JWS-0003}. Nonetheless, we leverage several \emph{sparsification} techniques in order to design distributed algorithms for general graphs. In particular, we first employ a deterministic \emph{multiplicative spanner} algorithm \cite{DBLP:conf/wdag/GhaffariK18,DBLP:conf/stoc/RozhonG20} to obtain a $\log n/ \log \log n$-approximation for the weighted \emph{all-pairs shortest paths} (APSP) problem in $\widetilde{\mathcal{O}}(\sqrt{n})$ rounds. For unweighted graphs we leverage the recent deterministic \emph{near-additive spanner} due to Elkin and Mater \cite{DBLP:conf/podc/ElkinM21}, leading to a $(1 + \epsilon)$-approximate algorithm for APSP, for any constant $\epsilon > 0$. Although this does not quite reach the performance of the state of the art algorithm of Kuhn and Schneider \cite{DBLP:conf/soda/AugustineHKSS20}, which yields an exact solution for weighted graphs with asymptotically the same round-complexity (modulo polylogarithmic factors), we stress that our algorithms are \emph{deterministic}. 

Moreover, we use \emph{cut sparsifiers} in order to provide near-optimal algorithms for any cut-related problem in $\widetilde{\mathcal{O}}(\sqrt{n})$ rounds, while we also observe an $\widetilde{\mathcal{O}}(\sqrt{n})$-round algorithm for determing the girth, i.e. the smallest cycle, using a standard trade-off between the girth and the number of edges. Here it is important to point out that our algorithmic scheme "Sparsify \& Conquer" is primarily meaningful when the output requirement is \emph{global}. For example, for the Min-Cut problem, if we require that every node knows a cut at the end of the distributed algorithm, we show an $\widetilde{\Omega}(\sqrt{n})$ round-complexity lower bound for any non-trivial approximation based on a technical lemma in \cite{DBLP:conf/soda/AugustineHKSS20}. However, in many settings this approach may disseminate an overly amount of information. Indeed, under the usual requirement that each node has to know its "side" on the cut, we establish exponentially faster algorithms.

\subsubsection{Simulating \texorpdfstring{$\congest$}{}-based Algorithms}

This accelerated algorithm for Min-Cut is obtained through a connection with the concept of \emph{low-congestion shortcuts}, due to Ghaffari and Haeupler \cite{DBLP:conf/soda/GhaffariH16}. Specifically, in this framework the performance-guarantee for a problem is parameterized in terms of the number of rounds required to solve the standard part-wise aggregation problem. A fascinating insight of Ghaffari and Haeupler \cite{DBLP:conf/soda/GhaffariH16} is that more "structured" topologies (e.g. planar graphs) enable faster algorithms for solving such problems, bypassing some notorious lower bounds under general graphs. Our observation is that a limited amount of global power, in the form of $\ncc$, interacts particularly well with this line of work since $\ncc$ offers very fast primitives for the part-wise aggregation problem. As a result, this connection leads to the following result:

\begin{theorem}
    There exists an $\mathcal{O}(\polylog(n))$-round algorithm for $(1 + \epsilon)$-approximate Min-Cut in $\congest + \ncc$.
\end{theorem}

Note that this guarantees applies even if local edges are restricted to transfer only $\mathcal{O}(\log n)$ bits per round, i.e. the local network is modeled with $\congest$ instead of $\local$. Another notable corollary of this connection is an approximate single-source shortest paths algorithm (\Cref{corollary:low_congestion-shortests_paths}) based on a result by Haeupler and Li \cite{DBLP:conf/wdag/HaeuplerL18}, coming close to the algorithm of Augustine et al. \cite{DBLP:conf/soda/AugustineHKSS20} under the substantially more powerful $\hybrid$. We also present another simulation argument, which in a sense eliminates the dependence of the performance of a $\congest$ algorithm on the hop-diameter through an appropriate augmentation of the graph with global edges (see \Cref{proposition:sim-diameter}).

\subsubsection{Distance Computation Tasks}

Finally, we focus on distance computation tasks, and in particular, the complexity of determining the \emph{radius} and the \emph{diameter} of the underlying graph---the smallest and the largest of the eccentricities respectively. For the former, we show the following result:

\begin{theorem}
    \label{theorem:two}
    For any $\epsilon \in (0, 1/2]$, determining a $(3/2 - \epsilon)$-approximation for the radius of a weighted graph with probability $2/3$ requires $\widetilde{\Omega}(n^{1/3})$ rounds of $\hybrid$. For unweighted graphs, determining the radius requires $\widetilde{\Omega}(n^{1/3})$ rounds of $\hybrid$.
\end{theorem}

This limitation applies for any randomized distributed algorithm even if we allow a substantial probability of failure (i.e. Monte Carlo algorithms), and/or public (common) randomness. We should point out that our lower bound for unweighted graphs matches the known upper bound for \emph{approximate} radius, as the authors in \cite{censorhillel2020distance} provide a $(1 + \epsilon)$-approximation for all the unweighted eccentricities in $\widetilde{\mathcal{O}}(n^{1/3})$ rounds of $\hybrid$, for any constant $\epsilon > 0$. Our theorem also supplements the hardness result of Kuhn and Schneider \cite{10.1145/3382734.3405719} who established analogous lower bounds for the diameter. 

More precisely, we give a suitable \emph{dense} gadget graph whose edges correspond to the input-strings of two players endeavoring to solve the set disjointness problem. Then, we show that there is a gap in the value of the radius depending on whether the input of the two players is \emph{disjoint}. Our construction uses a \emph{bit-gadget}, a component introduced in \cite{DBLP:conf/wdag/AbboudCK16} (see also \cite{10.5555/2722129.2722241}) in order to show a linear lower bound for determining the radius in $\congest$, even for sparse graphs. Nonetheless, our reduction has several differences given that the source of the communication bottleneck is quite different in $\congest$ (where it suffices to induce a bottleneck in the \emph{communication cut} between the two players) compared to a model with all-to-all communication. As a result, we first prove an $\widetilde{\Omega}(n)$ round-complexity lower bound for determining a $(4/3-\epsilon)$-approximation of the radius in the \emph{broadcast} variant of the \emph{congested clique} ($\bcc$), for any $\epsilon \in (0, 1/3]$, even for unweighted graphs; we consider this result to be of independent interest. Next, with minor modifications in the construction we show \Cref{theorem:two}. These results require simulation arguments, establishing that Alice and Bob can indeed employ (or simulate) the communication pattern of the distributed algorithm in order to solve the set disjointness problem. In this context, for the $\hybrid$ model we make use of the simulation argument of Kuhn and Schneider \cite{10.1145/3382734.3405719}.

Finally, for the weighted diameter the state of the art algorithm in $\hybrid$ simply performs a Dijkstra search from an arbitrary source node and returns as the estimation the eccentricity (i.e. the largest distance) of the source node \cite{censorhillel2020distance}; an application of the triangle inequality implies that this algorithm yields a $2$-approximation of the actual diameter. We make a step towards improving this approximation ratio. Specifically, we show that for graphs with small degrees ($\Delta = \mathcal{O}(\polylog n)$) we can obtain a $3/2$-approximation of the diameter with asymptotically the same round-complexity, namely $\widetilde{\mathcal{O}}(n^{1/3})$ rounds. This result is based on the sequential algorithm of Roditty and Vassilevska W. \cite{DBLP:conf/stoc/RodittyW13}. Our contribution is to establish that their algorithm can be substantially parallelized in $\hybrid$; this is shown by employing some machinery developed in \cite{censorhillel2020distance} for solving in parallel multiple single-source shortest paths problems. 

\begin{table}[ht]
\tiny
\caption{An overview of our main results. Here it is assumed that $\epsilon > 0$ is an arbitrarily small constant.}
\centering 
\resizebox{\textwidth}{!}{\begin{tabular}{c | c | c | c | c | c }
\toprule
 Problem & Variant & Approximation & Model & Complexity & Technique
\\ \midrule
\multicolumn{1}{c|}{\multirow{2}{*}{Deterministic APSP}} & Unweighted & $1 + \epsilon$ &
\multicolumn{1}{c|}{\multirow{2}{*}{$\hybrid$}} & \multicolumn{1}{c|}{\multirow{2}{*}{$\widetilde{\mathcal{O}}(\sqrt{n})$}} & \multicolumn{1}{c}{\multirow{2}{*}{Sparsification: \cite{DBLP:conf/stoc/RozhonG20,DBLP:conf/podc/ElkinM21}}} \\
& Weighted & $\log n /\log \log n$ & & & \\ \midrule
MST & Weighted & Exact & \multicolumn{1}{c|}{\multirow{3}{*}{$\congest + \ncc$}} & $\mathcal{O}(\log^2 n)$ & \multicolumn{1}{c}{\multirow{3}{*}{Shortcuts: \cite{DBLP:conf/soda/GhaffariH16,DBLP:conf/wdag/HaeuplerL18}}} \\
Min-Cut & Weighted & $1 + \epsilon$ & & $\mathcal{O}(\polylog (n))$ & \\
SSSP & Weighted & $\polylog(n)$ & & $\widetilde{\mathcal{O}}(n^{\epsilon})$ & \\ \midrule
\multicolumn{1}{c|}{\multirow{3}{*}{Radius}} & Unweighted & $4/3 - \epsilon$ & $\bcc$ & $\widetilde{\Omega}(n)$ & \multicolumn{1}{c}{\multirow{3}{*}{Set Disjointness: \cite{DBLP:journals/siamdm/KalyanasundaramS92,DBLP:conf/wdag/AbboudCK16,10.1145/3382734.3405719}}} \\
& Unweighted & Exact & \multicolumn{1}{c|}{\multirow{2}{*}{$\hybrid$}} & \multicolumn{1}{c|}{\multirow{2}{*}{$\widetilde{\Omega}(n^{1/3})$}} & \\
& Weighted & $3/2 - \epsilon$ & & & \\ \bottomrule
\end{tabular}}
\label{table:results}
\end{table}

\subsection{Related Work}

As we explained in our introduction, the $\hybrid$ model was only recently introduced by Augustine, Hinnenthal, Kuhn, Scheideler, and Schneider \cite{DBLP:conf/soda/AugustineHKSS20}. Specifically, they developed several useful communication primitives in order to tackle distance computation tasks; most notably, for the SSSP problem they established a $(1 + o(1))$-approximate solution in $\widetilde{\mathcal{O}}(n^{1/3})$ rounds, while they also presented an algorithm with round complexity $\widetilde{\mathcal{O}}(\sqrt{n})$ for approximately solving the weighted APSP problem with high probability.\footnote{We will say that an event holds \emph{with high probability} if it occurs with probability at least $1 - 1/n^c$ for some constant $c > 0$.} Their lower bound for the APSP problem was matched in a subsequent work by Kuhn and Schneider \cite{10.1145/3382734.3405719}, showing that $\widetilde{\mathcal{O}}(\sqrt{n})$ rounds suffice in order to \emph{exactly} solve APSP. They also presented an $\widetilde{\Omega}(n^{1/3})$ lower bound for determining the diameter based on a reduction from the two-party set disjointness problem. 

Moreover, Censor-Hillel et al. \cite{censorhillel2020distance} improved several aspects of the approach in \cite{DBLP:conf/soda/AugustineHKSS20}, showing how to exactly solve multiple SSSP problems in $\widetilde{\mathcal{O}}(n^{1/3})$ rounds; they also presented near-optimal algorithms for approximating all the eccentricities in the graph. For the approximate SSSP problem an improvement over the result in \cite{DBLP:conf/soda/AugustineHKSS20} was recently achieved by Censor-Hillel et al. \cite{DBLP:conf/spaa/Censor-HillelLP21}, obtaining a $(1 + \epsilon)$-approximate algorithm in $\widetilde{O}(n^{5/17})$ rounds of $\hybrid$, for a sufficiently small constant $\epsilon > 0$. More restricted families of graphs (e.g. very sparse graphs or \emph{cactus} graphs) were considered by Feldmann et al. \cite{feldmann2020fast}, establishing an exponential speedup over some of the previous results even though they modeled the local network via $\congest$, which is of course substantially weaker than $\local$. Finally, Götte et al. \cite{gotte2020timeoptimal} provided several fast hybrid algorithms for problems such as connected components, spanning tree, and the maximal independent set.

The \emph{node-capacitated clique} model ($\ncc$) was recently introduced in \cite{10.1145/3323165.3323195}; it constitutes a much weaker---and subsequently much more realistic---model than the \emph{congested clique} ($\congclique$) of Lotker et al. \cite{10.1145/777412.777428} in which every node can communicate with \emph{any} other node (instead of only $\mathcal{O}(\log n)$ other nodes in $\ncc$) with $\mathcal{O}(\log n)$-bit messages. Indeed, in $\congclique$ a total of $\widetilde{\Theta}(n^2)$ bits can be transmitted in each round, whereas in $\ncc$ the cumulative broadcasting capacity is only $\widetilde{\Theta}(n)$ bits; as evidence for the power of $\congclique$ we note that even slightly super-constant lower bounds would give new lower bounds in circuit complexity, as implied by a simulation argument in \cite{DBLP:conf/podc/DruckerKO13}. 

Reductions from communication complexity to distributed computing are by now fairly standard in the literature; see \cite{10.5555/795665.796477,10.1145/1993636.1993686,10.5555/2095116.2095207} and references therein. We also refer to \cite{10.5555/795665.796477,10.1145/1993636.1993686,10.5555/2095116.2095207} for reductions in the \emph{broadcast} variant of $\congclique$ where in each round every node can send the \emph{same} $\mathcal{O}(\log n)$-bit message to all the nodes. In particular, we should mention that the authors in \cite{10.5555/2095116.2095207} present several lower bounds for subgraph detection (see \cite{DBLP:conf/wdag/DolevLP12,DBLP:journals/dc/Censor-HillelKK19}), a problem which is studied in the present work as well; naturally, these lower bounds directly apply for $\ncc$. Our construction for the radius is inspired by the gadget in \cite{DBLP:conf/wdag/AbboudCK16}, wherein the authors showed near-linear lower bounds for determining the radius in $\congest$, even for sparse networks. Finally, we refer to \cite{10.1007/978-3-642-41527-2_1,10.1007/978-3-662-45174-8_30,10.1145/3212734.3212750} for some of the state of the art technology for the Min-Cut problem.

\section{Preliminaries}

We assume that the network consists of a set of $n$ communication entities (e.g. processors) with $[n] \define \{1, 2, \dots, n\}$ the set of IDs, and a local communication \emph{topology} given by a graph $G = (V, E, w)$. We will tacitly posit that $G$ is \emph{undirected}, unless explicitly stated otherwise; we also assume that for all $e \in E, w(e) \in \{1, 2, \dots, W\}$, for some $W = \poly (n)$. At the beginning each node knows the identifiers of each node in its  neighborhood, but has no further knowledge about the topology of the graph. Communication occurs in \emph{synchronous rounds}; in every round nodes have unlimited computational power\footnote{Nonetheless, we remark that most of our algorithms use a reasonable amount of computation.} to process the information they posses. The \emph{local} communication mode will be modeled with $\local$, for which in each round every node can exchange a message of \emph{arbitrary} size with its neighbors in $G$ via the \emph{local} edges. The \emph{global} communication mode uses $\ncc$ for which in each round every node can exchange $\mathcal{O}(\log n)$-bit\footnote{Our results could be parameterized by the size of the message $B$, but for concreteness we assume throughout this paper that $B = \mathcal{O}(\log n)$.} messages with up to $\mathcal{O}(\log n)$ arbitrary nodes via \emph{global} edges. More broadly, one can parameterize hybrid networks by the number of bits $\lambda$ that can be exchanged via local edges, and the number of bits $\gamma$ that can be exchanged via the global mode. Interestingly, all standard models can be seen as instances of this general parameterization; namely, $\local: \lambda = \infty, \gamma = 0, \congest: \lambda = \mathcal{O}(\log n), \gamma = 0, \congclique: \lambda = 0, \gamma = \mathcal{O}(n \log n),\footnote{This follows from Lenzen's routing \cite{DBLP:conf/podc/Lenzen13}.} \ncc: \lambda = 0, \gamma = \mathcal{O}(\log^2 n)$. 

If the capacity of some channel is exceeded the corresponding nodes will only receive an \emph{arbitrary} (potentially adversarially selected) subset of the information according to the capacity of the network, while the rest of the messages are dropped. The performance of a distributed algorithm is measured in terms of its \emph{round-complexity}---the number of rounds required so that every node knows its part of the output; for randomized protocols it will suffice to reach the desired state with high probability. Finally, all of the derived round-complexity upper bounds in $\hybrid$ should be thought of as having a minimum with the (hop) diameter of the network.

\subsection{Useful Communication Primitives}

A \emph{distributive aggregate function} $f$ maps a multiset $S = \{x_1, \dots, x_N\}$ of input values to some value $f(S)$, such that there exists an aggregate function $g$ so that for any multiset $S$ and any partition $S_1, \dots S_{\ell}, f(S) = g(f(S_1), \dots, f(S_{\ell}))$; typical examples that we will use include $\textsc{Max}, \textsc{Min}$, and $\textsc{Sum}$. Now consider that we are given a distributive aggregate function $f$ and a set $A \subseteq V$, so that every member of $A$ stores \emph{exactly} one input value. The \emph{aggregate-and-broadcast} problem consists of letting every node in the graph learn the value of $f$ evaluated at the corresponding input.

\begin{lemma}[\cite{10.1145/3323165.3323195}, Theorem 2.2]
    \label{lemma:AB}
    There exists an algorithm in $\ncc$ which solves the aggregate-and-broadcast problem in $\mathcal{O}(\log n)$ rounds. 
\end{lemma}

In the $(k, \ell)$-token dissemination problem (henceforth abbreviated as $(k,\ell)$-TD) there are $k$ (distinct) tokens (or messages), each of size $\mathcal{O}(\log n)$ bits, with every node initially having at most $\ell$ tokens. The goal is to guarantee that every node in the graph has collected all of the tokens.

\begin{lemma}[\cite{DBLP:conf/soda/AugustineHKSS20}, Theorem 2.1]
    \label{lemma:TD}
    There exists a randomized algorithm in $\hybrid$ which solves the $(k, \ell)$-TD problem on connected graphs in $\widetilde{\mathcal{O}}(\sqrt{k} + \ell)$ rounds with high probability.
\end{lemma}

Note that this round-complexity scales very favorably compared to the use of only one of the two communication modes comprising $\hybrid$. Indeed, even the \emph{gossip problem}---which corresponds to the $(n, 1)$-TD---requires $\widetilde{\Omega}(n)$ rounds in $\ncc$ (\cite{10.1145/3323165.3323195}), while in the $\local$ model we clearly require $\Omega(D)$ rounds. We will sometimes employ the following special case of \Cref{lemma:TD}, where recall that $\bcc$ stands for the \emph{broadcast} variant of $\congclique$.

\begin{corollary}
    \label{corollary:bcc}
We can simulate with high probability one round of $\bcc$ with $\widetilde{\mathcal{O}}(\sqrt{n})$ rounds of $\hybrid$.
\end{corollary}

\subsection{Communication Complexity}

Most of our lower bounds are established based on the communication complexity of \emph{set disjointness}, arguably the most well-studied problem in communication complexity (e.g., see \cite{DBLP:journals/toc/HastadW07,Nisan06thecommunication,10.1145/146637.146684}). More precisely, consider two communication parties---namely Alice and Bob---with infinite computational power. Every player is given a binary string of $k$-bits, represented with $x, y \in \{0,1\}^k$ respectively, and their goal is to determine the value of a function $f(x, y)$ by interchanging messages between each other. The players are allowed to use randomization, and the complexity is measured by the expected number of communication in the worst case \cite{10.1145/800135.804414}. For probabilistic protocols the players are required to give the right answer with some probability bounded away from $1/2$, i.e. to outperform random guessing; for concreteness, we assume that the probability of being correct should be $2/3$. It is also interesting to point out that common (public) randomness is allowed, with Alice and Bob sharing an infinite string of independent coin tosses.

In the set disjointness problem ($\disj_k$) the two parties have to determine whether there exists $i \in [k]$ such that $x_i = y_i = 1$; in other words, if the inputs $x$ and $y$ correspond to subsets of a universe $\Omega$, the problem asks whether the two subsets are disjoint---with a slight abuse of notation this will be represented with $x \cap y = \emptyset$. We will use the following celebrated result due to Kalyanasundaram and Schnitger \cite{DBLP:journals/siamdm/KalyanasundaramS92}.

\begin{theorem}[\cite{DBLP:journals/siamdm/KalyanasundaramS92}]
    \label{theorem:disj}
    The randomized communication complexity of $\disj_k$ is $\Omega(k)$.
\end{theorem}

\section{Sparsification in Hybrid Networks}

As a warm-up, we commence this section by presenting an $\widetilde{\mathcal{O}}(\sqrt{n})$ \emph{randomized} protocol for solving \emph{any} problem on \emph{sparse} graphs in the $\hybrid$ model. More importantly, we also present a \emph{deterministic} algorithm with asymptotically the same round complexity, up to polylogarithmic factors. Next, we present several applications of this result in general graphs via distributed \emph{sparsification} techniques.

\subsection{Randomized Protocol}

\begin{proposition}[Randomized Hybrid Algorithm for Sparse Networks] 
    \label{proposition:sparse}
Consider an $n$-node (connected) graph $G = (V, E, w)$. There exists a randomized algorithm so that every node in $V$ can learn the entire topology of the graph in $\widetilde{\mathcal{O}}(\sqrt{|E|})$ rounds of $\hybrid$ with high probability.
\end{proposition}

\begin{proof}
First, note that it suffices to solve an instance of the $(k,\ell)$-TD problem with $k = |E|$ and $\ell = \Delta$, where $\Delta$ denotes the maximum degree; indeed, every node $u \in V$ has to disseminate its adjacency list, consisting of $\deg(u)$ number of tokens, each of size $\mathcal{O}(\log n)$ bits. Yet, the token dissemination protocol of \Cref{lemma:TD} can only yield a round complexity of $\widetilde{\mathcal{O}}(\sqrt{|E|} + \Delta)$. We will show how to substantially accelerate this process and obtain the desired round-complexity.

As part of the first step, every node $u \in V$ has to transmit to the rest of the network its degree $\deg(u)$; this can be solved in $\widetilde{\mathcal{O}}(\sqrt{n})$ rounds by virtue of \Cref{corollary:bcc}. Next, we distinguish between the following two cases:

First, if $\Delta \leq \sqrt{|E|}$ the adjacency list of every node fits into at most $\sqrt{|E|}$ number of messages of $\mathcal{O}(\log n)$ bits. Thus, it suffices to employ the $(k, \ell)$-TD protocol of \Cref{lemma:TD} with $k = |E|$ and $\ell = \sqrt{|E|}$. It should be noted that every node can check that $\Delta \leq \sqrt{|E|}$ given that every degree was broadcast during the previous step.

Otherwise, assume that $\Delta > \sqrt{|E|}$. Let us denote with $Q = \{ u \in V : \deg(u) > \sqrt{|E|} \}$; again, note that every node in the graph knows this set by virtue of our previous step. Moreover, the handshaking lemma implies that $|Q| \leq 2 \sqrt{|E|}$. Our proposed algorithm proceeds in rounds, where in every iteration a \emph{single} node from $Q$ interacts with its neighbors in $G$; the order in which we process the nodes from $Q$ is assumed to be some fixed (predetermined) rule based on their IDs (e.g. ascending order), and importantly, we can guarantee synchronization as the IDs of the nodes in $Q$ are known to all the nodes. Now consider some iteration in which we process a node $u \in Q$. The main idea is to balance the \emph{load} in the neighborhood of $u$ via the local network. Specifically, given that $\deg(u) > \sqrt{|E|}$ and that the total load is $|E|$ messages, $u$ can redistribute this load among the nodes in $N(u) \cup \{u\}$, so that every node has at most $\sqrt{|E|}$ messages to broadcast. This can be performed in $2$ rounds of $\local$ (\Cref{fig:load_balancing}). After every such iteration, the number of nodes with load more than $\sqrt{|E|}$ decreases by at least one, and given that initially $|Q| \leq 2 \sqrt{|E|}$, it follows that after at most $2 \sqrt{|E|}$ rounds every node will have to broadcast at most $\sqrt{|E|}$ tokens. Finally, after balancing the load we can employ the $(k, \ell)$-TD protocol of \Cref{lemma:TD} with $k = |E|$ and $\ell =  \mathcal{O}(\sqrt{k})$, concluding the proof.
\end{proof}

\begin{figure}[!ht]
    \centering
    \includegraphics[scale=0.5]{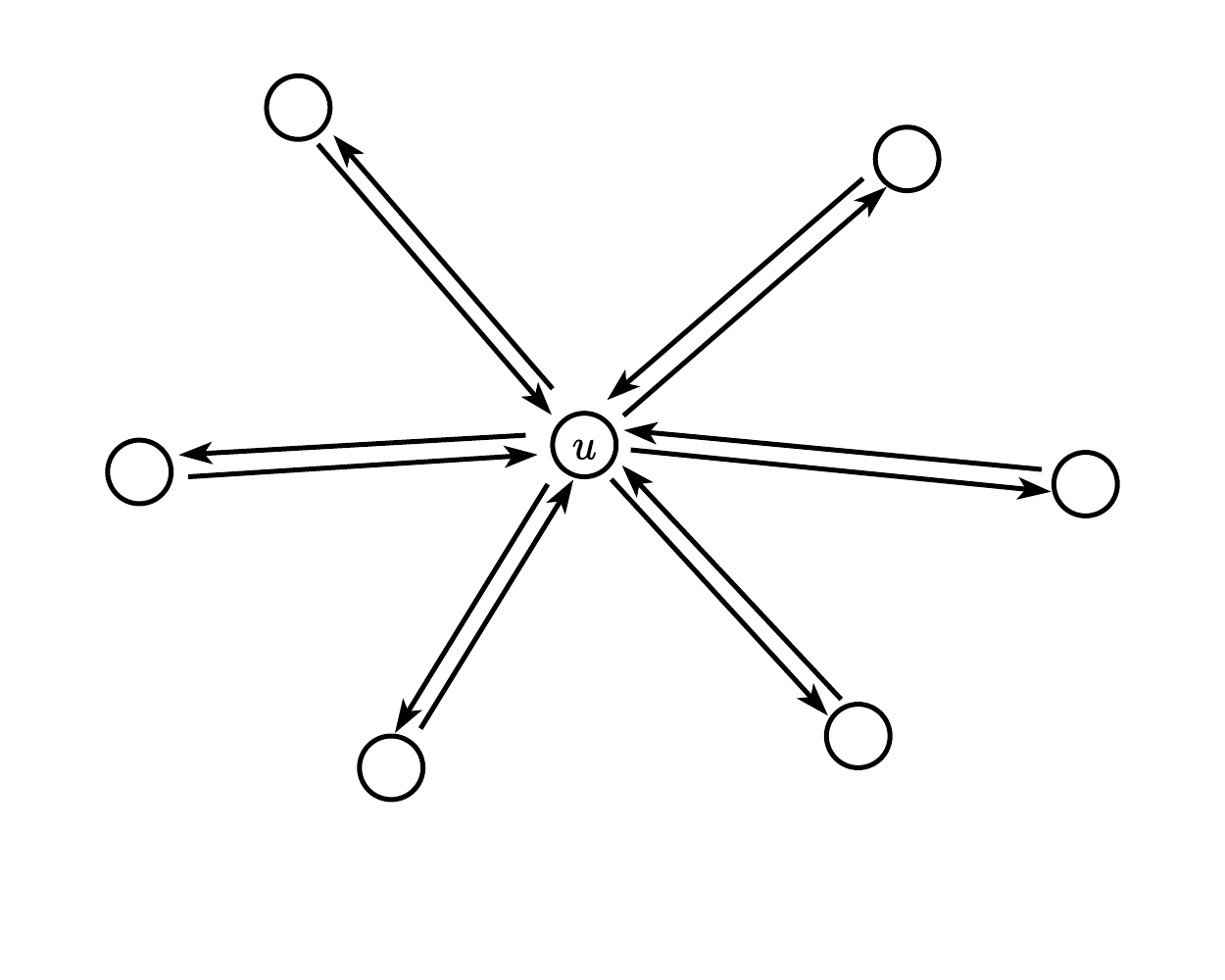}
    \caption{Load balancing in the neighborhood $N(u)$ of node $u$. In the first round, every node $v \in N(u)$ sends all of its tokens to $u$ via the local network; then, $u$ splits the total load \emph{uniformly} (but otherwise arbitrarily), and in the second round it transmits to every node $v \in N(u)$ its corresponding load (set of tokens) again via the local network.}
    \label{fig:load_balancing}
\end{figure}

\subsection{Deterministic Protocol}

Before we proceed with our deterministic algorithm let us first recall that the \emph{strong diameter} of a subset $C \subset V$ is the diameter of the subgraph induced by $C$; in contrast, the \emph{weak diameter} of $C$ is measured in the original graph. We will analyze and explain every step of the algorithm separately. We stress that the round-complexity in some steps has not been optimized since it would not alter the asymptotic running time of the protocol. Also note that in the sequel we use the words component and cluster interchangeably.

\begin{algorithm}[H]
\SetAlgoLined
\textbf{Input}: An $n$-node graph $G = (V, E)$ such that $|E| = \widetilde{\mathcal{O}}(n)$.\\
\textbf{Output Requirement}: Every node knows the entire topology of $G$.
\begin{enumerate}
    \item[\textit{1.}] Determine a partition of the nodes $V$ into $\mathcal{C}_1, \mathcal{C}_2, \dots, \mathcal{C}_k$ via the \\Garay-Kutten-Peleg algorithm such that for all $i$,
 \begin{itemize}
     \item[(i)] the strong diameter of $\mathcal{C}_i$ is $\mathcal{O}(\sqrt{n})$;
     \item[(ii)] $|\mathcal{C}_i| \geq \sqrt{n}$.
 \end{itemize}
    \item[\textit{2.}] Let $C_1, C_2, \dots, C_N = \textsc{Fragment}(\mathcal{C}_1, \dots, \mathcal{C}_k)$ such that for all $i$,
\begin{itemize}
    \item[(i)] the weak diameter of $C_i$ is $\mathcal{O}(\sqrt{n})$;
    \item[(ii)] $ \sqrt{n} \leq |C_i| < 2 \sqrt{n}$. 
\end{itemize}
    \item[\textit{3.}] Broadcast the IDs of the components' leaders.
    \item[\textit{4.}] Distribute all the $C_i$'s via the $\textsc{MatchingComponents}$ subroutine.
    \item[\textit{5.}] Assign every edge on a component, and perform $\textsc{LoadBalancing}$.
    \item[\textit{6.}] Disseminate all the information.
\end{enumerate}
 \caption{Deterministic $\hybrid$ Algorithm for Sparse Networks}
\end{algorithm}

\textbf{Step} $1$. The first step of the algorithm partitions the set of nodes into a collection of connected components $\mathcal{C}_1, \dots, \mathcal{C}_k$, so that the minimum size is at least $\sqrt{n}$ and the strong diameter in every component is $\mathcal{O}(\sqrt{n})$; for simplicity we will assume that $\sqrt{n}$ is an integer. This step will be implemented with the standard Garay-Kutten-Peleg ($\GKP$) algorithm \cite{DBLP:journals/siamcomp/GarayKP98,Lenzen}. Specifically, $\GKP$ is an MST algorithm which operates in two phases; we will only need the first phase. The main idea is to gradually perform merges but in a "balanced" manner. More precisely, $\GKP$ maintains a set of components. In each iteration $i$ every component with diameter at most $2^i$ determines the minimum-weight outgoing edge, which is subsequently added to a set of "candidates" edges. Then, the algorithm determines a maximal matching on this set, updating the components accordingly. If a component with diameter smaller than $2^i$ did not participate in the maximal matching, the algorithm automatically incorporates the edge that was selected by it. This process is repeated for $i = 0, 1, \dots, \lceil \log \sqrt{n} \rceil$, leading to a partition of $V$ into $\mathcal{C}_1, \dots, \mathcal{C}_k$.

\begin{lemma}[\cite{Lenzen}]
    \label{lemma:GKP}
    At the end of the first phase of the $\GKP$ algorithm every component has strong diameter $\mathcal{O}(\sqrt{n})$, while every component has at least $\sqrt{n}$ nodes.
\end{lemma}

This lemma verifies our initial claim for step $1$. Moreover, note that the first phase of $\GKP$ can be implemented in $\mathcal{O}(\sqrt{n} \log^* n)$ rounds in $\congest$; naturally, in $\hybrid$ we can substantially reduce the number of rounds, but this would not affect the overall asymptotic complexity as there is an inherent bottleneck in subsequent steps of the algorithm. 

\textbf{Step} $2$. The $\textsc{Fragment}$ subroutine of the second step is particularly simple. If a component $\mathcal{C}_i$ is such that $|\mathcal{C}_i| < 2 \sqrt{n}$ it remains intact. Otherwise, the component $\mathcal{C}_i$ is decomposed arbitrarily into disjoint fragments each of size between $\sqrt{n}$ and $2 \sqrt{n}$. Let $C_1, \dots, C_N$ be the induced partition of $V$. By virtue of \Cref{lemma:GKP} we know that the \emph{weak} diameter of every $C_i$ is $\mathcal{O}(\sqrt{n})$, although note that the induced graph on $C_i$ is potentially disconnected. This step is made to ensure that the components have roughly the same size, while it can be trivially implemented in $\mathcal{O}(\sqrt{n})$ rounds of $\local$.

\textbf{Step} $3$. We assume that every component has elected a leader, e.g. the node with the smallest ID. There are overall $N \leq \sqrt{n}$ IDs to be broadcast to the entire graph. This can be implemented with $N$ (deterministic) broadcasts in $\ncc$, which requires $\mathcal{O}(\sqrt{n} \log n)$ rounds.

\textbf{Step} $4$. The purpose of this step is to ensure that every node knows the composition---i.e. the set of IDs---of every other cluster. To this end, the leader of every component $C_i$ selects arbitrarily $N-1 \leq |C_i| - 1$ \emph{representative} nodes from $C_i$, and devises a (bijective) mapping from these nodes to all the other components; the leader also informs via the local network the corresponding nodes. Then, the protocol proceeds in rounds: In every iteration a single component interacts with all the others, and specifically, every representative node sends its ID to the leader of its assigned component. This is repeated for all the components, and after $\mathcal{O}(\sqrt{n})$ rounds every representative node will be matched with some node on its corresponding component; see \Cref{fig:sub1}. Having established this matching every node can disseminate through the global network the IDs of all the nodes in its own component to its assigned node. This process requires at most $2 \sqrt{n}$ rounds since every component has size less than $2 \sqrt{n}$, and every node participates in at most one matching. Finally, the composition (the set of IDs) of every component is revealed to each node after $\mathcal{O}(\sqrt{n})$ additional rounds of the local network. 

\begin{figure}[!ht]
\centering
\begin{subfigure}{.5\textwidth}
  \centering
  \includegraphics[scale=0.35]{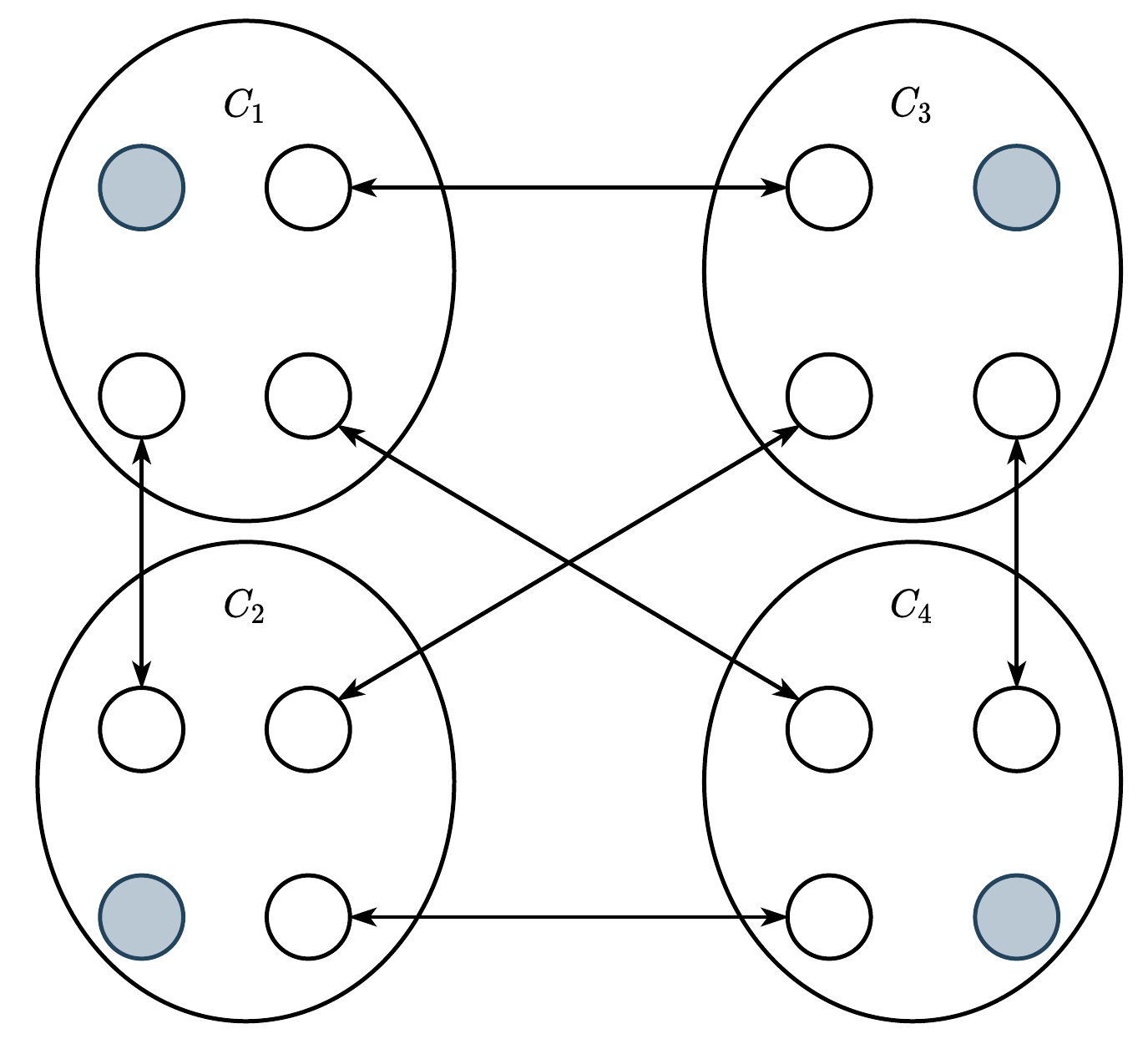}
  \caption{An example of $\textsc{MatchingComponents}$.}
  \label{fig:sub1}
\end{subfigure}%
\begin{subfigure}{.5\textwidth}
  \centering
  \includegraphics[scale=0.4]{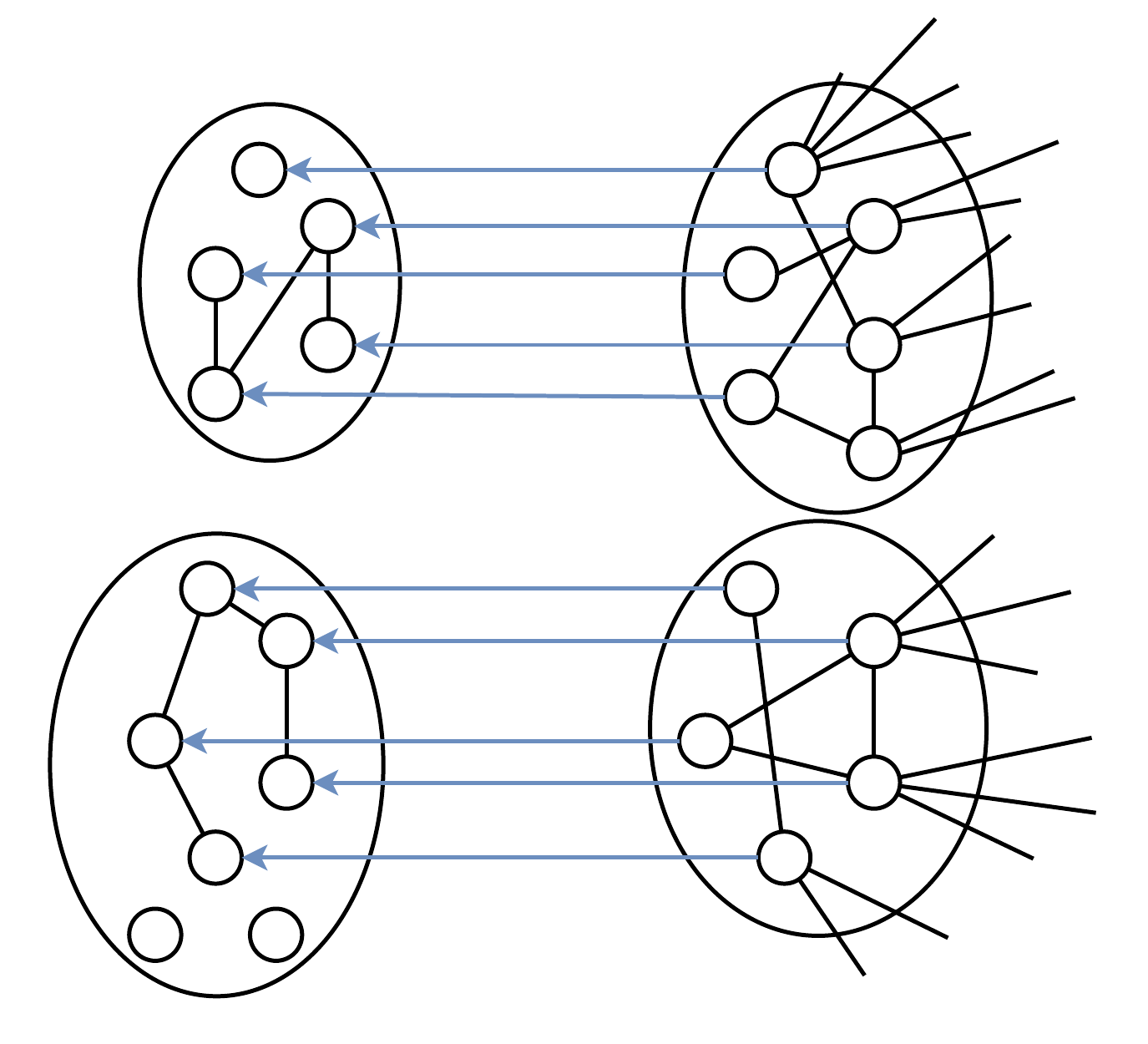}
  \caption{$\textsc{LoadBalancing}$: Transferring load from overloaded to underloaded components. We have highlighted with blue "global" edges.}
  \label{fig:sub2}
\end{subfigure}
\end{figure}

\textbf{Step} $5$. First of all, every edge with incident nodes residing on the same component is assigned to the component of its endpoints. Otherwise, the edge is assigned to one of the components according to some deterministic rule; e.g. the component with the smaller ID. In this context, the \emph{load} of every component is the number of edges it has to disseminate. Notice that the nodes of each component can learn every component's load in $\mathcal{O}(\sqrt{n})$ rounds. Initially, the load of each component is distributed uniformly within the nodes of the component, which requires $\mathcal{O}(\sqrt{n})$ rounds. The $\textsc{LoadBalancing}$ mechanism works as follows: It splits the components into a set of "overloaded" components with load more than $2 |E|/N$, and a set of "underloaded" components with load less than $|E|/N$; every other component does not need any further processing. In every iteration we map arbitrarily (e.g. the overloaded component with the smallest ID is mapped to the underloaded component with the smallest ID, and so on) overloaded components to underloaded ones, so that the mapping is one-to-one and maximal. Then, every assigned overloaded component transmits as much load as is required to its corresponding component via the global network (see \Cref{fig:sub2}) until one of the two becomes balanced---according to the previous notion. This can be performed in $\widetilde{\mathcal{O}}(1)$ by virtue of Step $4$ (recall that $|E| = \widetilde{O}(n)$). Then, we remove any components that have been balanced and we proceed recursively for the remaining ones. It is easy to see that this process requires at most $N$ iterations since in every iteration we eliminate at least one component from requiring further balancing, while at the end of this step every component will have $\widetilde{\mathcal{O}}(\sqrt{n})$ load.

\textbf{Step} $6$. The final step is fairly straightforward. First, observe that a single component can transfer its entire load to another component in $\widetilde{\mathcal{O}}(1)$ rounds via the global network; this follows because (i) every component has $\widetilde{\mathcal{O}}(\sqrt{n})$ load due to the load balancing step, and (ii) every component has by construction roughly $\sqrt{n}$ nodes. Assume that the components $C_1, \dots, C_N$ are sorted in ascending order with respect to their IDs. Then, at iteration $i$ component $C_j$ transfers its load to $C_{r+1}$, where $r = i + j \mod N$. This is repeated for $i = 0, 1, \dots, N-2$. It is easy to see that this deterministic protocol guarantees that (i) no \emph{collisions} occur, and (ii) every component eventually receives the load from all other components. As we previously argued every such iteration requires $\widetilde{\mathcal{O}}(1)$ rounds in $\hybrid$. Thus, overall this step requires $\widetilde{\mathcal{O}}(\sqrt{n})$ rounds, leading to the following conclusion:

\deterministic*

More broadly, our deterministic protocol can be used for any $m$-edge (connected) graph by forming clusters of size $\Theta(\sqrt{m})$ nodes, so that every node learns the topology after $\widetilde{\mathcal{O}}(\sqrt{m})$ rounds. Consequently, this leads to a derandomization of the token dissemination algorithm of Augustine et al. \cite{DBLP:conf/soda/AugustineHKSS20} in the regime $k \geq n$:

\begin{proposition}
    \label{proposition:derandomization}
    There exists a deterministic algorithm which solves the $(k, \ell)$-TD problem in $\widetilde{\mathcal{O}}(\sqrt{k})$ rounds of $\hybrid$, assuming that $k \geq n$.
\end{proposition}

\subsection{Distributed Sparsification Techniques}

Before we present applications of our protocol for general graphs, we first review some basic sparsification techniques. The goal is to efficiently sparsify the graph in a distributed fashion, while approximately preserving some \emph{structure} in the graph. We will use two fundamental notions of sparsifiers.

\subsubsection{Spanners}

The first structure one may wish to retain with sparsification is every pairwise distance in the graph. To this end, we will employ the notion of a graph \emph{spanner}, a fundamental object in graph theory with numerous applications in distributed computing \cite{DBLP:journals/siamcomp/PelegU89a}. To be more precise, for a graph $G = (V, E)$, a subgraph $H$ is an $\alpha$-\emph{stretch spanner} if every pairwise distance in $H$ is at most an $\alpha \geq 1$ factor larger than the distance in $G$, without ever underestimating; i.e., for all $u, v \in G, d_G(u, v) \leq d_H(u, v) \leq \alpha \cdot d_G(u, v)$. Naturally, we desire spanners with small stretch and a limited number of edges. It is well-known that any $n$-node graph admits a $(2k-1)$-stretch spanner with $\mathcal{O}(n^{1 + 1/k})$ number of edges, while this trade-off is optimal conditioned on Erd\H{o}s girth conjecture. In the distributed context, we will use the following result of Rozhon and Ghaffari:

\begin{theorem}[\cite{DBLP:conf/stoc/RozhonG20}]
    \label{theorem:spanners}
    Consider an $n$-node weighted graph $G = (V, E)$. There exists a deterministic distributed algorithm in $\congest$ which computes a $(2k-1)$-stretch spanner of size $\widetilde{\mathcal{O}}(k n^{1 + 1/k})$ in $\polylog(n)$ rounds.
\end{theorem}

Besides \emph{multiplicative} spanners, we will also use \emph{near-additive} spanners. More precisely, a subgraph $H$ of $G$ is an $(\alpha, \beta)$-stretch spanner if for all $u, v \in G, d_G(u, v) \leq d_H(u, v) \leq \alpha \cdot d_G(u, v) + \beta$; for $\beta = 0$ this recovers the previous notion of a multiplicative spanner. Moreover, for $\alpha = 1 + \epsilon$, for an arbitrarily small $\epsilon > 0$, the spanner is called near-additive. In this context, we will leverage the following recent result due to Elkin and Matar:

\begin{theorem}[\cite{DBLP:conf/podc/ElkinM21}]
    \label{theorem:near_additive}
    Consider an $n$-node unweighted graph $G = (V, E)$. For any constants $\epsilon \in (0, 1)$ and $\rho \in (0, 1/2)$, there is an algorithm in $\congest$ which computes a $(1 + \epsilon, \beta)$-stretch spanner of $G$ with $\widetilde{\mathcal{O}}(n)$ number of edges after $\mathcal{O}(\beta n^{\rho})$ rounds, where $\beta = \mathcal{O}((\log \log n/\rho + 1/\rho^2)^{\log \log n + 1/\rho})$.
\end{theorem}

\subsubsection{Cut Sparsifiers}

Another fundamental class of sparsifiers endeavors to approximately preserve the weight of every \emph{cut} in the graph. Recall that for a subset of vertices $S \subset V$ we define 

\begin{equation}
    \cut_G(S) = \sum_{u \in S, v \in V \setminus S} w(u, v).
\end{equation}

To this end, we will employ the sparsification algorithm developed by Koutis \cite{koutis2014simple}. We should remark that the algorithm of Koutis actually returns a \emph{spectral} sparsifier, which is a strictly stronger notion than a cut sparsifier \cite{DBLP:journals/cacm/BatsonSST13}, but we will not use this property here.

\begin{theorem}[\cite{koutis2014simple}, Theorem 5]
    \label{theorem:koutis}
    Consider a graph $G = (V, E, w)$. There exists a distributed algorithm in $\congest$ such that for any $\epsilon > 0$ outputs a graph $H = (V, \widehat{E}, \widehat{w})$ after $\widetilde{\mathcal{O}}(1/\epsilon^2)$ rounds such that (i) $(1 - \epsilon) \cut_H(S) \leq \cut_G(S) \leq (1 + \epsilon) \cut_H(S)$ for any $S \subset V$, and (ii) the expected number of edges in $H$ is $\widetilde{\mathcal{O}}(n/\epsilon^2)$.
\end{theorem}

\subsection{Applications}

\subsubsection{Deterministic APSP}
\label{subsubsubsection:apsp}

In the $\alpha$-approximate \emph{all-pairs shortest paths} problem every node $u \in V$ has to learn a value $d'(u, v)$ such that $d(u, v) \leq d'(u, v) \leq \alpha \cdot d(u, v)$, for all $v \in V$. In this context, we establish the following result:  

\begin{proposition}[Weighted APSP]
    \label{proposition:apsp}
    Consider an $n$-node weighted graph $G = (V, E)$. There exists a deterministic $\log n/\log \log n$-approximation algorithm for the APSP problem which runs in $\widetilde{\mathcal{O}}(\sqrt{n})$ rounds of $\hybrid$.
\end{proposition}

\begin{proof}
First, we use the local network in order to implement the algorithm of Rozhon and Ghaffari with $k = \mathcal{O}(\log n/\log \log n)$, yielding a $\log n/\log \log n$-stretch spanner $H$ such that $|E_H| = \widetilde{\mathcal{O}}(n)$; notice that \Cref{theorem:spanners} implies that this step can be implemented in $\polylog(n)$ rounds. Next, we use \Cref{theorem:deterministic_sparse} so that every node learns the subgraph $H$ in $\widetilde{\mathcal{O}}(\sqrt{n})$ rounds of $\hybrid$, and the theorem follows given that $H$ is a $\log n/\log \log n$-stretch spanner.  
\end{proof}

For unweighted graphs we will use near-additive spanners (\Cref{theorem:near_additive}) to improve upon the approximation ratio established for weighted graphs.

\begin{proposition}[Unweighted APSP]
    Consider an $n$-node unweighted graph $G = (V, E)$. For any constant $\epsilon \in (0, 1)$, there exists a deterministic $(1 + \epsilon)$-approximation algorithm for the APSP problem which runs in $\widetilde{\mathcal{O}}(\sqrt{n})$ rounds of $\hybrid$.
\end{proposition}

\begin{proof}
Let $\beta$ be defined as in \Cref{theorem:near_additive}, where we take $\rho < 1/2$. First, each node in the network will determine all the nodes which lie within $\beta/\epsilon$ distance. This can be trivially performed in $\lceil \beta/\epsilon \rceil$ rounds of $\local$ since the graph is unweighted. Next, we leverage the deterministic algorithm of \Cref{theorem:near_additive} to compute a $(1 + \epsilon, \beta)$-stretch spanner $H$ in $o(\sqrt{n})$ rounds. Afterwards, we use our deterministic protocol of \Cref{theorem:deterministic_sparse} so that every node in the graph learns the spanner $H$ after $\widetilde{O}(\sqrt{n})$ rounds; here we used that the number of edges in the spanner is $\widetilde{O}(n)$. Now consider two nodes $u, v \in V$ such that $d_G(u,v) \geq \beta/\epsilon$.
Then, since $H$ constitutes a $(1 + \epsilon, \beta)$-stretch spanner it follows that $d_G(u, v) \leq d_H(u, v) \leq (1 + \epsilon) d_G(u, v) + \beta \leq (1 + \epsilon) d_G(u, v) + \epsilon d_G(u, v) \leq (1 + 2\epsilon) d_G(u, v)$. Otherwise, if $d_G(u, v) < \beta/\epsilon$, then both nodes know the exact distance from each other by virtue of our previous local step. Thus, we have recovered a $(1 + 2\epsilon)$-approximation for the APSP problem, as desired. 
\end{proof}

\subsubsection{Cut Problems}

Moreover, we will leverage the distributed algorithm of Koutis in order to obtain efficient algorithms for cut-related problems in the $\hybrid$ model. We wish to convey the robustness of our approach by presenting a guarantee for a series of cut problems. First, we recall the following: The \emph{minimum cut} problem consists of identifying a partition of the vertices $V$ into $S$ and $V \setminus S$ in order to minimize the weight of $\cut_G(S)$; it admits an efficient centralized solution, for example, via Karger's celebrated algorithm \cite{10.5555/313559.313605}. Note that for an unweighted graph the minimum cut coincides with the \emph{edge connectivity}. The \emph{$s-t$ minimum cut} problem is similar to the minimum cut problem, but the nodes $s$ and $t$ are restricted to reside on different sets of the partition; see \cite{DBLP:conf/stoc/BenczurK96,10.1145/1993636.1993674}. Finally, in the \emph{sparsest cut} problem we are searching for a partition $S, V \setminus S$ that minimizes the quantity $\cut_G(S)/(|S| \cdot | V \setminus S|)$; it is known that the sparsest cut problem is ${\mathcal{NP}}$-hard \cite{10.1145/1502793.1502794,10.1145/100216.100257}.

\begin{proposition}[$\hybrid$ Algorithms for Cut Problems]
    \label{proposition:cuts}
    For any $n$-node graph $G = (V, E, w)$ and for any $\epsilon \in (0, 2)$, we can compute with high probability a $(1 + \epsilon)$-approximation in expected $\widetilde{\mathcal{O}}(\sqrt{n}/\epsilon + 1/\epsilon^2)$ rounds of $\hybrid$ for the following problems: (i) the minimum $s-t$ cut, (ii) the minimum cut, and (iii) the sparsest cut.
\end{proposition}

\begin{proof}
First, we apply the sparsification algorithm of Koutis \cite{koutis2014simple}, employing only the local network for $\widetilde{\mathcal{O}}(1/\epsilon^2)$ rounds in order to identify a subgraph $H$. \Cref{theorem:koutis} implies that with high probability $(1 - \epsilon) \cut_H(S) \leq \cut_G(S) \leq (1 + \epsilon) \cut_H(S)$, for all $S \subset V$. Then, we leverage our algorithm of \Cref{proposition:sparse} so that every node knows the entire topology of $H$ after expected $\widetilde{\mathcal{O}}(\sqrt{|E_H|})$ rounds of the $\hybrid$ model; $|E_H|$ represents the number of edges of the cut sparsifier $H$, and \Cref{theorem:koutis} implies that $\E[|E_H|] = \widetilde{\mathcal{O}}(n/\epsilon^2)$. Thus, the expected number of rounds for executing our algorithm from \Cref{theorem:deterministic_sparse} is $\widetilde{\mathcal{O}}(\E[\sqrt{|E_H|}]) = \widetilde{\mathcal{O}}(\sqrt{\E[|E_H|]}) = \widetilde{\mathcal{O}}(\sqrt{n}/\epsilon)$ (Cauchy-Schwarz). Then, every node can determine locally the solution to the corresponding problem in the sparsified graph $H$, which also yields a $(1 + 4\epsilon)$-approximation if $\epsilon \in (0,1/2)$; finally, rescaling $\epsilon$ concludes the proof.
\end{proof}

Naturally, our approach yields results for other cut-related problems, such as computing a $(1 - \epsilon)$-approximate maximum cut, or determining an approximate \emph{Gomory-Hu} tree \cite{10.2307/2098881}. It should be noted that the local computation required to compute \emph{exactly} the sparest cut---even in the spectral sparsifier---is most likely exponential; one could employ the $\mathcal{O}(\sqrt{\log n})$-approximation algorithm of Arora et al. \cite{10.1145/1502793.1502794} in order to reduce the local computation, sacrificing analogously the approximation ratio guarantee.

\paragraph{Lower Bound}

This approach is meaningful for cut-related problems once we impose a stronger output requirement. Namely, we guarantee that every node will know at the end of the distributed algorithm the entire composition of an approximate cut. In fact, for such an output requirement we can establish an almost-matching lower bound:

\begin{proposition}
    \label{proposition:min_cut-lower_bound}
    Determining a $W/n$-approximation for the minimum cut problem requires $\widetilde{\Omega}(\sqrt{n})$ rounds of $\hybrid$, where $W \geq n$ is the maximum edge-weight, assuming that every node has to know a cut at the end of the distributed algorithm.
\end{proposition}

The approximation ratio here measures the (multiplicative) discrepancy between the derived solution and the minimum (weighted) cut in the graph. Our approach is based on a technique developed in \cite{DBLP:conf/soda/AugustineHKSS20} for establishing lower bounds in the $\hybrid$ model. Specifically, they showed how to induce an information bottleneck for a certain class of graphs. Formally, we paraphrase their main technical lemma:

\begin{lemma}[\cite{DBLP:conf/soda/AugustineHKSS20}, Lemma 4.4]
    \label{lemma:hybrid-lower_bound}
Consider an $n$-node graph $G = (V, E)$ consisting of a subgraph $G'=(V',E')$, and a path of length $L$ edges from some node $a \in V'$ to $b\in V \setminus V'$ such that $a$ is the unique node from $V'$ in the path. If the nodes in $V'$ are given by the state of some random variable $X$, and node $b$ needs to learn the realization of $X$, then every randomized algorithm which solves the problem in the $\hybrid$ model requires $\Omega\left(\min\{L, H(X)/(L \log^2 n) \} \right)$ rounds.
\end{lemma}

Note that $H(X)$ represents the (Shannon) entropy of random variable $X$. Armed with this lemma, we are ready to construct a "hard" instance for the Min-Cut problem. 

\begin{proof}[Proof of \Cref{proposition:min_cut-lower_bound}]
Consider two nodes $a, b$ and a path of length $L = \lfloor \sqrt{n} \rfloor$ edges connecting them, such that the weight of every edge is $W \geq n$. Moreover, we let $V' = V_1' \cup V_2'$ such that $V_1' \cap V_2' = \emptyset$, with $a \in V_1'$ and some node $u \in V_2'$. Now every other node in $V'$ is assigned to one of $V_1'$ and $V_2'$ (exclusively) based on the outcome of an unbiased random coin. Finally, we connect node $u$ to every node in $V_1'$ with edges of unit-weight, while each of $V_1'$ and $V_2'$ are interconnected via simple paths with edge-weights $W$, as illustrated in the \Cref{fig:min_cut}. Observe that for $W \geq n$, the minimum cut of the induced graph is (independently from the random realization) the set $\{ \{u, v\} : v \in V_1' \}$, while every other cut of the graph yields an approximation ratio of at least $W/n$. As a result, $b$ has to know the entire set $V_1'$, with the exception of node $a$, in order to determine a reasonable approximation. Let $X$ be a random variable that encodes the IDs of the nodes in $V_1'$, excluding node $a$. It follows that $X$ is uniformly distributed over the subsets of $V' \setminus \{a, u\}$, implying that $H(X) = \log 2^{|V'| - 2} = \Omega(n)$ bits. Thus, given that node $b$ has to know the state of $X$ in order to determine the minimum cut, \Cref{lemma:hybrid-lower_bound} implies an $\widetilde{\Omega}(\sqrt{n})$ round-complexity lower bound. To be more precise, if $b$ does not know the state of $X$, it has to be assumed that the value of the minimum cut is $W$, leading to an approximation ratio larger than $W/n$.
\end{proof}

\begin{figure}[!ht]
    \centering
    \includegraphics[scale=0.5]{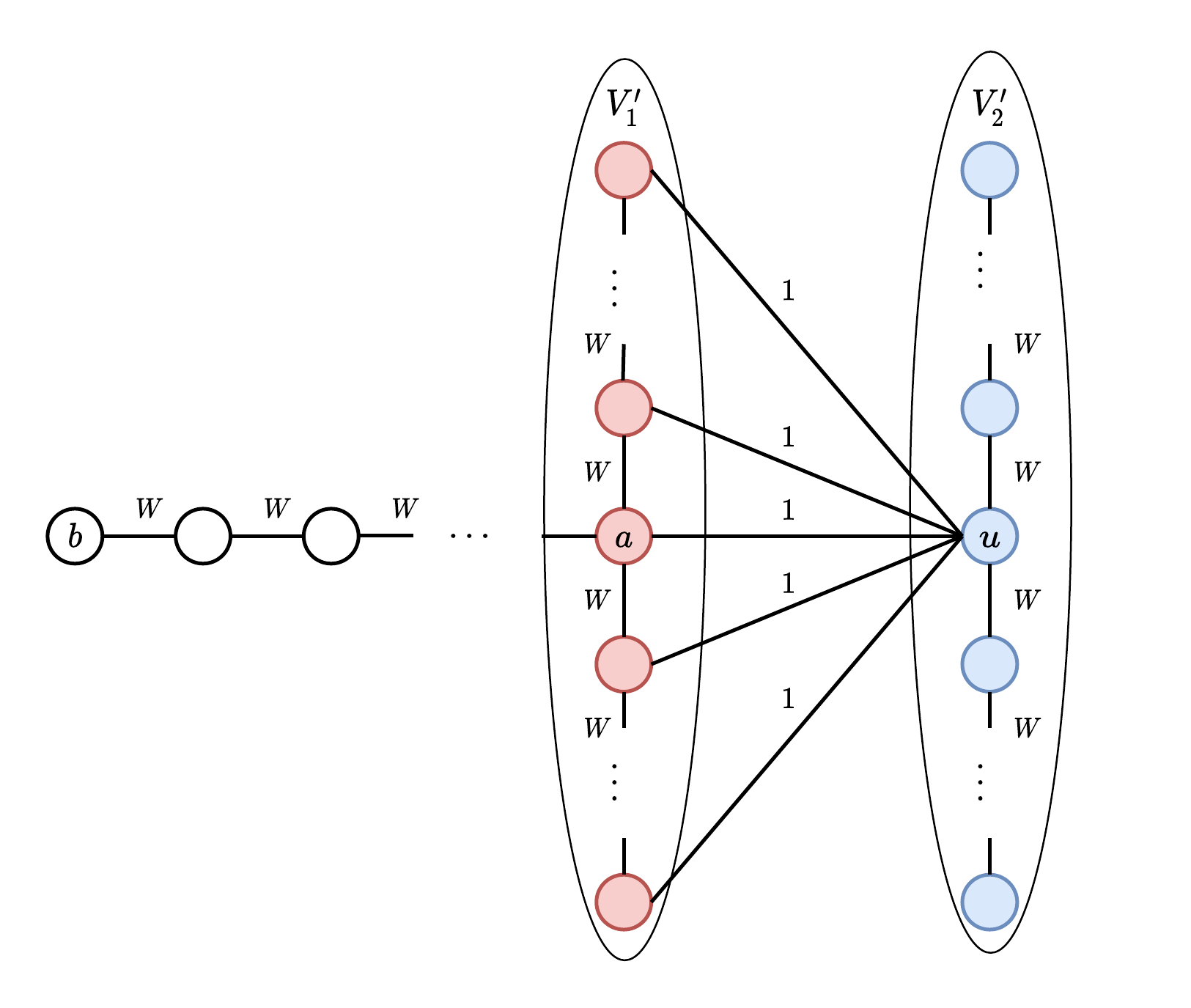}
    \caption{A hard Min-Cut instance in the $\hybrid$ model.}
    \label{fig:min_cut}
\end{figure}

\subsubsection{Girth}

Here we present an algorithm for determining the \emph{girth} of the graph; recall that the girth is defined as the length of the smallest cycle. We commence with the following standard lemma, establishing a trade-off between the girth and the number of edges in the graph.

\begin{lemma}[\cite{10.5555/581165}]
    \label{lemma:girth}
A graph of girth $g$ has at most $n^{1 + 1/\lfloor (g-1)/2 \rfloor} + n$ edges.
\end{lemma}

In particular, this lemma implies that if $g = \Omega(\log n)$, the underlying graph should be sparse. We leverage this observation with \Cref{theorem:deterministic_sparse} to establish the following result:

\begin{proposition}
    \label{proposition:girth}
There exists a deterministic algorithm for determining the girth of any graph $G$ in $\widetilde{\mathcal{O}}(\sqrt{n})$ rounds of $\hybrid$.
\end{proposition}

\begin{proof}
First, for $\log n$ rounds we let the nodes perform flooding through the local network; in this way, every node can determine the length of the smallest cycle it participates in, or $+\infty$ if no such cycle exists. Then, we employ the \emph{aggregate-and-broadcast} protocol of \Cref{lemma:AB} for the distributive aggregate function $\textsc{Min}$, where the input of every node corresponds to the number previously determined via flooding. If the result is not $+ \infty$ the algorithm terminates with every node knowing the girth of the graph. Otherwise, it follows that $g \geq 2 \log n$, where $g$ represents the girth. Thus, \Cref{lemma:girth} implies that $|E| = \mathcal{O}(n)$, and we can employ \Cref{theorem:deterministic_sparse} so that every node can determine the girth after $\widetilde{\mathcal{O}}(\sqrt{n})$ rounds.
\end{proof}

Again, we actually establish a much stronger result: every node can learn the entire composition of the minimum length cycle in $\widetilde{\mathcal{O}}(\sqrt{n})$, not just its length. We should note that for such an output requirement we can establish a matching lower bound similarly to the communication bottleneck induced for the Min-Cut problem (\Cref{proposition:min_cut-lower_bound}). However, if it suffices to let every node know the length of the minimum cycle it is unclear how to provide a meaningful lower bound. Indeed, when $g = \Omega(\log n)$ the underlying graph is sparse\footnote{Interestingly most of the upper bounds in the $\ncc$ \cite{10.1145/3323165.3323195} depend on the \emph{arboricity} of the graph, which roughly speaking is a measure of its sparsity; see the Nash-Williams theorem \cite{https://doi.org/10.1112/jlms/s1-39.1.12}.} and the standard approach---which is based on reducing the set disjointness problem on a suitably constructed instance---appears to fail in this case. In light of this we only give a lower bound of $\widetilde{\Omega}(n^{1/3})$ rounds for approximating the girth for \emph{directed} graphs; the following construction serves as a warm-up for our result in \Cref{section:radius}.

\paragraph{Lower Bound for Directed Girth}

We commence by providing a lower bound for $\bcc$; then, we will explain how to modify our construction for $\hybrid$. Specifically, we will present a reduction from the two-party set disjointness problem to approximating the directed girth. To this end, consider a set of nodes $U \cup V \cup V' \cup U'$, where we let $U = \{u_0, u_1, \dots, u_{k-1}\}, V = \{v_0, v_1, \dots, v_{k-1}\}, V' = \{v_0', v_1', \dots, v_{k-1}'\}$, and $U' = \{u_0', u_1', \dots, u_{k-1}'\}$. Moreover, we add the set of edges $\{ (v_i, v_i') : i \in [k]^* \} \cup \{ (u_i', u_i) : i \in [k]^* \}$, where $[k]^* \define \{0, 1, \dots, k-1\}$. Now assume that $x, y \in \{0,1\}^{k^2}$ represent the input strings of Alice and Bob respectively. We assume that Alice and Bob encode their inputs as edges on the graph, such that $(u_i, v_j) \in E \iff x_{i, j} = 1$, and $(v_j', u_i') \in E \iff y_{i, j} = 1$; in words, Alice encodes her input as edges between the nodes in $U$ and $V$, while Bob encodes his input as edges between the nodes in $V'$ and $U'$. We let $G_k^{x,y}$ represent the induced graph. This construction is illustrated in \Cref{fig:directed_girth}. 

\begin{figure}[!ht]
    \centering
    \includegraphics[scale=0.6]{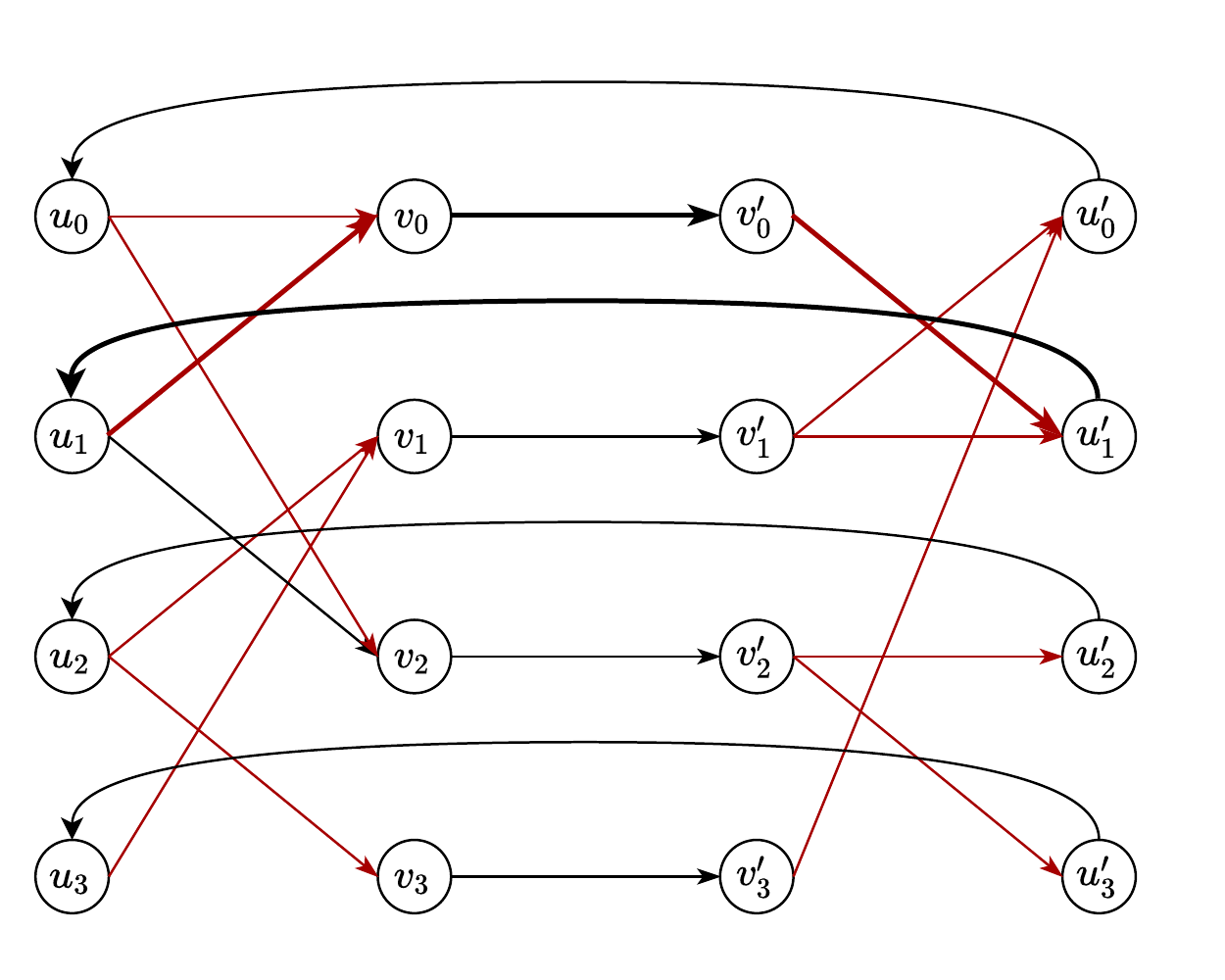}
    \caption{Reducing set disjointness to approximating the directed girth. Red edges correspond to the players' input strings. In this instance there exists a $4$-cycle in the induced graph (we have highlighted the corresponding edges), implying that the players' input strings are \emph{not} disjoint. }
    \label{fig:directed_girth}
\end{figure}

\begin{claim}
    \label{claim:directed_girth-bcc}
The girth $g$ of the directed graph $G_k^{x,y}$ is $4$ if $x \cap y \neq \emptyset$; otherwise, $g \geq 8$.
\end{claim}

\begin{proof}
First of all, if $x \cap y \neq \emptyset$ it follows that $x_{i, j} = y_{i, j} = 1$, for some $i, j \in [k]^*$. By construction, this implies that $(u_i, v_j) \in E \land (v_j', u_i') \in E$, and hence, there exists a cycle of length $4$; namely, $u_i \rightarrow v_j \rightarrow v_j' \rightarrow u_i' \rightarrow u_i$. Otherwise, observe that if $g < + \infty$, it must be that $4 \mid g$. As a result, it suffices to show that if $x \cap y = \emptyset$ there does not exist a $4$-cycle. Indeed, for the sake of contradiction posit the $4$-cycle $u_i \rightarrow v_j \rightarrow v_j' \rightarrow u_i' \rightarrow u_i$ for some $i, j \in [k]^*$ (observe that every $4$-cycle can be expressed in this form). This implies that $(u_i, v_j) \in E \land (v_j', u_i') \in E$, which in turn gives us that $x_{i, j} = y_{i, j} = 1$, contradicting the assumption that $x \cap y = \emptyset$.
\end{proof}

We assume that an $\alpha$-approximation algorithm for the (directed) girth should return a number $\widetilde{g}$ such that $g \leq \widetilde{g} \leq \alpha \cdot g$, where $g$ represents the actual girth of the directed graph. With that in mind, we are ready to establish the following:

\begin{theorem}
    \label{theorem:directed_girth-bcc}
    For any $\epsilon \in (0, 1]$, determining a $(2-\epsilon)$-approximation of the directed girth with probability $2/3$ requires $\widetilde{\Omega}(n)$ rounds of $\bcc$.
\end{theorem}

\begin{proof}
Consider an algorithm $\mathcal{A}$ in the $\bcc$ which determines a $(2-\epsilon)$-approximation of the directed girth with probability $2/3$. \Cref{claim:directed_girth-bcc} implies that Alice and Bob can employ algorithm $\mathcal{A}$ on graph $G_k^{x,y}$ in order to solve the set disjointness problem with a substantial probability. Indeed, Alice and Bob can directly simulate the communication protocol dictated by $\mathcal{A}$. As a result, we know from \Cref{theorem:disj} that $\Omega(k^2)$ bits have to be exchanged between the nodes of $U \cup V$ and $U' \cup V'$ during the execution of algorithm $\mathcal{A}$. However, this requires $\Omega(k^2/(k\log n)) = \widetilde{\Omega}(n)$ rounds given that every node can only transmit $\mathcal{O}(\log n)$ (distinct) bits per round in $\bcc$. 
\end{proof}

Next, we will show how to modify our construction in order to obtain an $\widetilde{\Omega}(n^{1/3})$ round-complexity lower bound in $\hybrid$. The main idea is to introduce some "gap" between the nodes in the graph which correspond to different players. To be precise, consider some parameter $\ell \in \mathbb{N}$; instead of connecting the nodes from $V$ to $V'$ and from $U'$ to $U$ directly via edges, we will introduce paths of length $\ell$ (edges). Notice that for $\ell = 1$ we recover our previous construction. Moreover, the encoding of the players' inputs $x,y \in \{0,1\}^{k^2}$ will remain exactly the same. We let $G_{k, \ell}^{x, y}$ be the induced graph. Similarly to \Cref{claim:directed_girth-bcc}, we can establish the following:

\begin{claim}
    \label{claim:directed_girth-hybrid}
The girth $g$ of the directed graph $G_{k, \ell}^{x, y}$ is $2 + 2\ell$ if $x \cap y \neq \emptyset$; otherwise, $g \geq 4 + 4\ell$.
\end{claim}

\begin{theorem} 
    \label{theorem:directed_girth-hybrid}
    For any $\epsilon \in (0,1]$, determining a $(2-\epsilon)$-approximation of the directed girth with probability $2/3$ requires $\widetilde{\Omega}(n^{1/3})$ rounds of $\hybrid$.
\end{theorem}

\begin{proof}
Consider an algorithm $\mathcal{A}$ in $\hybrid$ which determines a $(2-\epsilon)$ approximation of the directed girth with probability $2/3$. We know from \cite[Lemma 7.3]{10.1145/3382734.3405719} that Alice and Bob can together simulate $\lfloor \ell/2 \rfloor$ rounds of algorithm $\mathcal{A}$ on the graph $G_{k, \ell}^{x,y}$, while exchanging information only about messages from the global network. Thus, both Alice and Bob have determined a $(2-\epsilon)$-approximation of the girth of $G_{k, \ell}^{x,y}$, and \Cref{claim:directed_girth-hybrid} in turn implies that they have solved the set disjointness problem with probability $2/3$. As a result, Alice and Bob have exchanged $\Omega(k^2)$ bits (\Cref{theorem:disj}) during the simulation, implying that $\Omega(k^2)$ bits have been transmitted via the global network. This would require $\Omega(k^2/(n \log^2 n))$ rounds given that every node can only disseminate $\log^2 n$ bits per round via the global network. Overall, we have shown that algorithm $\mathcal{A}$ requires $\Omega(\min \{ \ell, k^2/(n \log^2 n)\})$; given that $k \times \ell = \Theta(n)$ this quantity is maximized for $k = \Theta(n^{2/3})$ and $\ell = \Theta(n^{1/3})$, concluding the proof.
\end{proof}

It should also be noted that a similar construction yields an $\widetilde{\Omega}(n^{1/3})$ round-complexity lower bound for determining the weight of the minimum-weight cycle for undirected graphs. In \Cref{appendix:detecting-counting} we provide very fast and simple algorithms for detecting and counting subgraphs, which are central problems in the realm of distributed computing (e.g., see \cite{DBLP:conf/wdag/DolevLP12}). In contrast, both communication modes comprising $\hybrid$ fail to provide efficient distributed algorithms. Indeed, in $\local$ we trivially require $\Omega(n)$ rounds when the (hop) diameter $D = \Theta(n)$; more interestingly, the same lower bound---up to polylogarithmic factors---applies for $\ncc$ as well. To be more precise, we show via a reduction from the set disjointness problem that counting $r$-cycles with $r =\mathcal{O}(1)$ requires $\widetilde{\Omega}(n)$ rounds even in the \emph{broadcast} variant of the congested clique ($\bcc$), which is substantially more powerful than $\ncc$ (up to a logarithmic factor). Similar limitations have been shown for \emph{detecting} certain classes of subgraphs in $\bcc$ \cite{DBLP:conf/podc/DruckerKO13}, but with some subtleties; see our discussion in \Cref{appendix:detecting-counting}.

\section{Simulating \texorpdfstring{$\congest$}{} Algorithms}
\label{section:simulation}

The approach developed in the previous section is primarily meaningful for problems with a very demanding output requirement; for example, APSP, or cut-related problems for which nodes have to learn the exact composition of the cut. In contrast, in this section we will show that, for the minimum cut problem, we can obtain substantially faster algorithms when the nodes have to simply learn their "side" on the cut, which constitutes the usual output requirement in distributed algorithms. This result (\Cref{corollary:low_congestion-min_cut}) will be established through a connection with the concept of \emph{low-congestion shortcuts}, which also implies other important results as well; e.g. for the SSSP problem (\Cref{corollary:low_congestion-shortests_paths}). We also present another simulation argument, leading to an accelerated algorithm for approximating the diameter. It should be stress that for this section we model the local network via the weaker $\congest$ model.

\subsection{Low-Congestion Shortcuts}

Consider a graph $G = (V, E)$ under the $\congest$ model, and a partition of $V$ into $k$ parts $P_1, \dots, P_k$ such that the induced graph $G[P_i]$ is connected. A recurring scenario in distributed algorithms consists of having to perform simultaneous aggregations in each part; this will be referred to as the \emph{part-wise aggregation} problem. For example, an instance of this problem corresponds to determining the minimum-weight outgoing edge in the context of Boruvka's celebrated algorithm. A very insightful observation by Ghaffari and Haeupler \cite{DBLP:conf/soda/GhaffariH16} was to parameterize the performance of algorithms based on the complexity of the part-wise aggregation problem. For instance, if it admits a solution in $Q$ rounds, under any collection of parts, we can compute an MST in $\mathcal{O}(Q \log n)$ rounds via Boruvka's algorithm. Now although in general graphs $Q = \mathcal{O}(\sqrt{n} + D)$ rounds, with the bound being existential tight for certain topologies, a key insight of Ghaffari and Haeupler \cite{DBLP:conf/soda/GhaffariH16} is that special classes of graphs allow for accelerated algorithms via \emph{shortcuts}; most notably, for planar graphs they showed that the part-wise aggregation problem can be solved in $\widetilde{\mathcal{O}}(D)$ rounds of $\congest$, bypassing the notorious $\Omega(\sqrt{n})$ rounds for "global" problems under general graphs. 

In the hybrid model this connection is particularly useful since the $\ncc$ model enables very fast algorithms for solving the part-wise aggregation problem:

\begin{lemma}[\cite{10.1145/3323165.3323195}]
    \label{lemma:part-wise-ncc}
    The part-wise aggregation problem admits a solution with high probability in $\mathcal{O}(\log n)$ rounds in $\ncc$.
\end{lemma}

As a result, we can directly derive a near-optimal algorithm for the minimum spanning tree problem through an implementation based on Boruvka's algorithm:

\begin{corollary}
    There exists a distributed algorithm which computes with high probability an MST in $\mathcal{O}(\log^2 n)$ rounds of $\congest + \ncc$.
\end{corollary}

We refer to \Cref{appendix:MST} for the detailed implementation. It should be noted that a deterministic $\mathcal{O}(\log^2 n)$ algorithm in $\congest + \ncc$ for the MST problem was developed in \cite{feldmann2020fast} with very different techniques. More importantly, Ghaffari and Haeupler \cite{DBLP:conf/soda/GhaffariH16} managed to establish the following:

\begin{theorem}[\cite{DBLP:conf/soda/GhaffariH16}]
    If we can solve the part-wise aggregation problem in $Q$ rounds of $\congest$, there exists an $\widetilde{\mathcal{O}}(Q \poly(1/\epsilon))$ distributed algorithm for computing with high probability a $(1 + \epsilon)$-approximation of the minimum cut, for any sufficiently small $\epsilon > 0$.
\end{theorem}

We should remark that in \cite{DBLP:conf/soda/GhaffariH16} the authors establish this result only for planar graphs, but their argument can be directly extended in the form of this theorem. As a result, if we use \Cref{lemma:part-wise-ncc} we arrive at the following conclusion:

\begin{corollary}
    \label{corollary:low_congestion-min_cut}
    Consider any $n$-node weighted graph. There exists a $\mathcal{O}(\polylog (n))$-round algorithm in $\congest + \ncc$ for computing with high probability a $(1 + \epsilon)$-approximation of Min-Cut, for any sufficiently small constant $\epsilon > 0$.
\end{corollary}

In terms of \emph{exact} Min-Cut, one can obtain an $\widetilde{\mathcal{O}}(\sqrt{n})$-round algorithm in $\congest + \ncc$ by simulating the recent algorithm due to Dory et al. \cite{DBLP:conf/stoc/DoryEMN21}, which requires $\widetilde{\mathcal{O}}(D + \sqrt{n})$ rounds of $\congest$; an analogous simulation argument is employed in the next subsection, so we omit the proof here. However, this leaves a substantial gap between exact and approximate Min-Cut. Moreover, analogous results can be established for computing approximate shortest paths by virtue of a result by Haeupler and Li \cite{DBLP:conf/wdag/HaeuplerL18}: 

\begin{corollary}
    \label{corollary:low_congestion-shortests_paths}
    Consider any $n$-node weighted graph. There exists with high probability a $\polylog (n)$-approximate algorithm for the single-source shortest paths problem which runs in $\widetilde{O}(n^{\epsilon})$ rounds in $\congest + \ncc$, for any constant $\epsilon > 0$.
\end{corollary}

We remark that Haeupler and Li \cite{DBLP:conf/wdag/HaeuplerL18} actually provide a more general result, but we state this special case for the sake of simplicity. Of course, there are other applications as well, as we have certainly not exhausted the literature. Overall, this connection illustrates another very concrete motivation of low-congestion shortcuts.

\subsection{Diameter}

We also provide another notable simulation argument. In particular, the main idea is to augment the local topology with a limited number of "global" edges so that the resulting graph has a small diameter, and at the same time the solution to the underlying problem remains invariant.

\begin{proposition}
    \label{proposition:sim-diameter}
    There exists a distributed algorithm in $\congest + \ncc$ which determines a $3/2$-approximation of the diameter in $\mathcal{O}(\sqrt{n \log n})$ rounds with high probability.
\end{proposition}

\begin{proof}
First of all, we know that there exists a distributed algorithm by Holzer et al. \cite{DBLP:conf/wdag/HolzerPRW14} which computes with high probability a $3/2$-approximation of the diameter in $\mathcal{O}(\sqrt{n \log n} + D)$ rounds of $\congest$. We will show how to simulate this algorithm on a "virtual" graph. Specifically, consider a graph $\widehat{G}$ which derives from $G$ via the following augmentation: We consider an arbitrary balanced binary tree on the $n$ nodes, and every edge in the tree which is not present in the local topology dictated by $G$ is incorporated into $\widehat{G}$ with weight $3/2 \times n W$; every other edge will be included in $\widehat{G}$ with the same weight. By construction, observe that a $3/2$-approximation of the diameter in $\widehat{G}$ also serves as a $3/2$-approximation of the diameter in $G$. Moreover, we can simulate any $\congest$ algorithm in $\widehat{G}$ since the "virtual" tree has constant max-degree, and as such, the communication on top of these edges can be implemented via the $\ncc$ model. As a result, the claim follows since the diameter of $\widehat{G}$ is $\mathcal{O}(\log n)$.
\end{proof}

\section{Distance Computations}
\label{section:radius}

\subsection{Lower Bound for the Radius}

In this subsection we show a lower bound of $\widetilde{\Omega}(n^{1/3})$ rounds for computing the radius---the smallest eccentricity of the graph---in the $\hybrid$ model, even for unweighted graphs. We commence by constructing a suitable "gadget" in the $\bcc$ model, and then we will massage it appropriately to establish a guarantee for $\hybrid$ as well. Our construction is inspired by that in \cite{DBLP:conf/wdag/AbboudCK16} which established a sharp lower bound for sparse networks in $\congest$. First, consider a set of nodes $U \cup V \cup V' \cup U'$, and we let $U = \{u_0, u_1, \dots, u_{k-1}\}, V = \{v_0, v_1, \dots, v_{k-1}\}, V' = \{v_0', v_1', \dots, v_{k-1}'\}$, and $U' = \{u_0', u_1', \dots, u_{k-1}'\}$. We also incorporate edges of the form $\{v_i, v_i'\}$ for all $i \in \{0, 1, \dots, k-1\} = [k]^*$. An important additional ingredient is the \emph{bit-gadget} \cite{DBLP:conf/wdag/AbboudCK16}, which works as follows: every node in $U$ and $U'$ will inherit some edge-connections based on the binary representation of their index; the role of this component will become clear as we proceed with the construction. Formally, consider a (different) set of nodes $F, T, F', T'$, and let $F = \{f_0, f_1, \dots, f_{l-1}\}, T = \{t_0, t_1, \dots, t_{l-1}\}, F' = \{f_0', f_1', \dots, f_{l-1}'\}$, and $T' = \{t_0', t_1', \dots, t_{l-1}'\}$, where $l = \lceil \log k \rceil$. Now consider a node $u_i \in U$, and let $i = \overline{b_{l-1} \dots b_1 b_0}$ be the binary representation of its index; for every $j \in [l]^*$ we add the edge $\{u_i, f_j\}$ if $b_j = 0$; otherwise, we add the edge $\{u_i, t_j\}$. This process is also repeated for the nodes in $U'$ (with respect to the sets $F'$ and $T'$). Next, we add the edges $\{f_j, t_j\}$ and $\{f_j', t_j'\}$ for all $j \in [l]^*$, while a critical element of the construction is the set of edges $\{ \{f_j, t_j'\} : j \in [l]^* \} \cup \{ \{t_j, f_j'\} : j \in [l]^* \}$. We also incorporate in the graph two nodes $w, w'$ such that $w$ is connected to all the nodes in $U$ and $V$, and $w'$ is connected to all the nodes in $U'$ and $V'$. Finally, we add three nodes $z_0, z_1, z_2$, as well as the set of edges $ \{ \{z_0, z_1\} \} \cup \{ \{z_1, z_2\} \} \cup \{ \{z_0, u_i \} : i \in [k]^* \}$.

Having constructed this base graph the next step is to encode the input of Alice and Bob as edges on the induced graph. Specifically, let $x \in \{0,1\}^{k^2}$ and $y \in \{0,1\}^{k^2}$ represent the input strings of Alice and Bob respectively. We let $x_{i, j} = 1 \iff \{u_i, v_j\} \in E$, and $y_{i, j} = 1 \iff \{v_j', u_i'\} \in E$. We denote the induced graph with $G_k^{x, y}$; an example of our construction is illustrated in \Cref{fig:radius}.

\begin{figure}[!ht]
    \centering
    \includegraphics[scale=0.6]{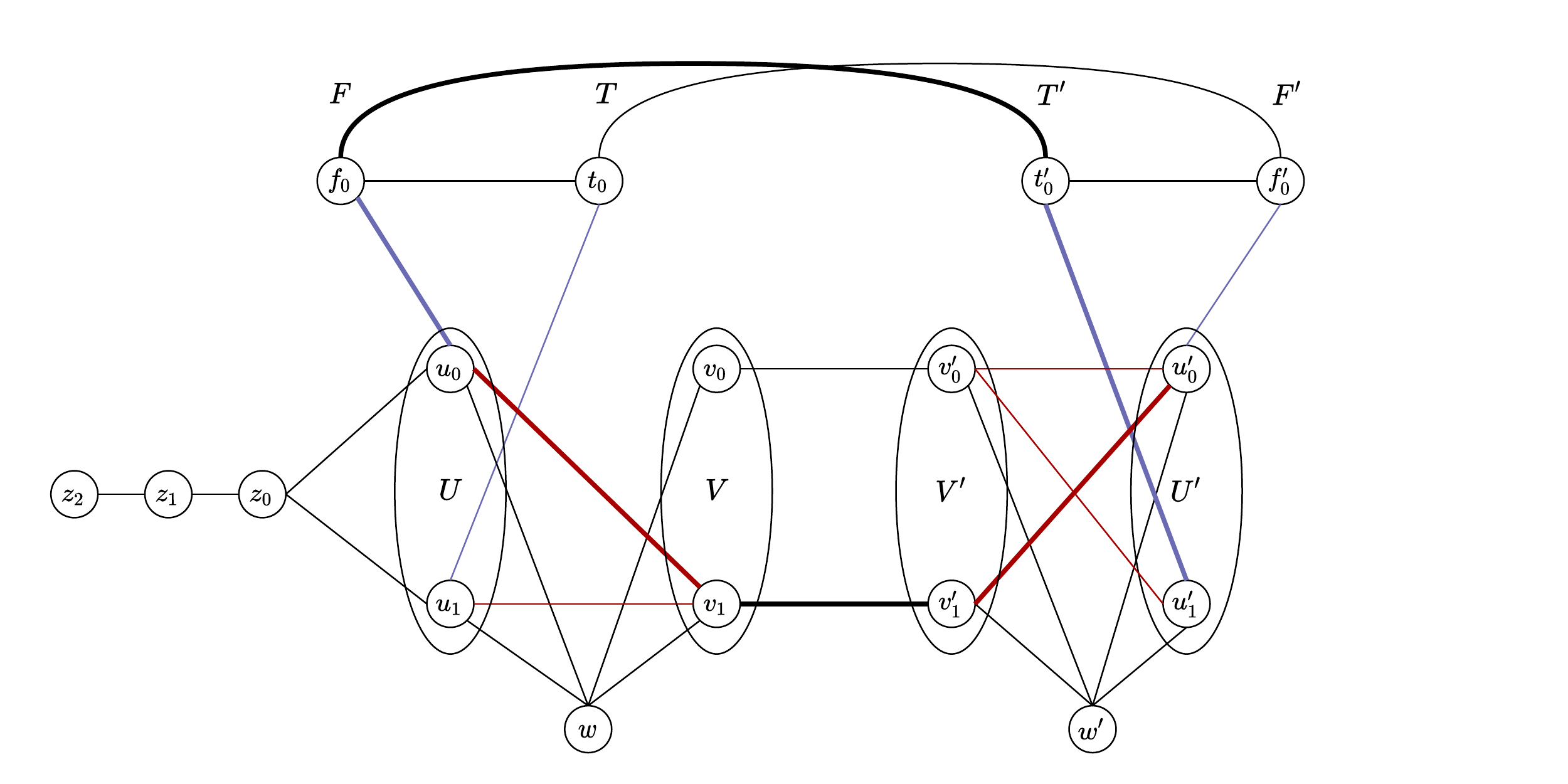}
    \caption{An example of our construction for the radius. The red edges correspond to the players' inputs, while the blue edges map nodes of $U$ and $U'$ to their \emph{bit-gadget}. Observe that $d(u_0, u_0') = 3$ (we have highlighted the corresponding path in the figure) as $\{u_0, v_1 \} \in E$ and $\{v_1', u_0'\} \in E$, implying that $x \cap y \neq \emptyset$. We have also highlighted the path of length $3$ from $u_0$ to $u_1'$ through the bit-gadget. }
    \label{fig:radius}
\end{figure}

\begin{claim}
\label{claim:aux_1}
For every node $u$ in $G_k^{x,y}$ besides the nodes in $U$ it follows that $\ecc(u) \geq 4$.
\end{claim}

\begin{proof}
First, consider some node $u \notin U \cup \{z_0, z_1, z_2\}$; it is easy to see that $d(u, z_2) \geq 4$, implying that $\ecc(u) \geq 4$ and $\ecc(z_2) \geq 4$. Moreover, it follows that $d(z_1, w') \geq 5$ and $d(z_0, w') \geq 4$, concluding the proof.
\end{proof}

\begin{claim}
    \label{claim:aux_2}
The radius $R$ of $G_k^{x,y}$ is $3$ if $x \cap y \neq \emptyset$; otherwise, $R = 4$.
\end{claim}

\begin{proof}
First, assume that $x \cap y \neq \emptyset$; in particular, let $x_{i, j} = y_{i, j} = 1$ for some $i, j \in [k]^*$. We will show that $\ecc(u_i) = 3$, which in turn implies that $R = 3$ given that $\ecc(u) \geq 3$ for all $u$. Indeed, observe that $d(u_i, z_0) = 1, d(u_i, z_1) = 2, d(u_i, z_2) = 3$; moreover, $d(u_i, w) = 1$, and through $w$ node $u_i$ can reach all the nodes in $U$ and $V$ in only two steps; this also implies that $d(u_i, v_j') = 3, \forall p \in [k]^*$. In addition, observe that $d(u_i, f_p) \leq 2$ and $d(u_i, t_p) \leq 2$ for all $p \in [l]^*$; thus, it also follows that $d(u_i, f_p') \leq 3$ and $d(u_i, t_p') \leq 3$ for all $p \in [l]^*$. The next step is to show that $d(u_i, u_p') = 3$ if $i \neq p$. Indeed, given that $i \neq p$ it follows that there exists some bit in the binary representation of their indexes in which the two numbers disagree; let $r$ be that position, and assume without any loss of generality that $i$ has a bit $1$, while $j$ has a bit $0$ in the $r$-th position. This implies that there exists the path $u_i \rightarrow t_r \rightarrow f_r' \rightarrow u_p'$, and hence, $d(u_i, u_p') = 3$. Finally, given that $x_{i, j} = y_{i, j} = 1$, we can deduce that $d(u_i, u_i') = 3$, as there exist edges $\{ u_i, v_j \}$ and $\{ v_j', u_i' \}$ (also notice that $d(u_i, w') = 3$).

In contrast, if $x \cap y = \emptyset$ it is easy to see that for all $i \in [k]^*, d(u_i, u_i') = 4$, implying along with \Cref{claim:aux_1} that $R = 4$.
\end{proof}

\begin{theorem}
    \label{theorem:radius-bcc}
    For any $\epsilon \in (0, 1/3]$, determining a $(4/3-\epsilon)$-approximation for the radius of an unweighted graph with probability $2/3$ requires $\widetilde{\Omega}(n)$ rounds of $\bcc$.
\end{theorem}

Next, we will show how to adapt this construction for the $\hybrid$ model. Specifically, if $\ell \in \mathbb{N}$ is some parameter, we introduce the following modifications: instead of connecting the corresponding nodes in $V$ with $V'$, $F$ with $T'$, and $F'$ with $T$ directly via edges, we will connect them via paths of length $\ell$ edges; moreover, we create the path $z_0 \rightarrow z_1 \rightarrow \dots \rightarrow z_{\ell +1}$ (in place of $z_0 \rightarrow z_1 \rightarrow z_2$). As before, the players' inputs $x, y \in \{0,1\}^{k^2}$ shall be encoded as edges between $U$ with $V$, and $V'$ with $U'$ for Alice and Bob respectively. Let $G_{k,\ell}^{x,y}$ be the induced graph; the following claim admits an analogous proof to \Cref{claim:aux_2}:

\begin{claim}
    \label{claim:radius-hybrid}
The radius $R$ of $G_{k, \ell}^{x,y}$ is $\ell + 2$ if $x \cap y \neq \emptyset$; otherwise, $R \geq \ell + 3$.
\end{claim}

Consequently, we are ready to state the implied lower bound in the $\hybrid$ model for unweighted graphs.

\begin{theorem} 
    \label{theorem:radius-hybrid}
    Determining the radius of an unweighted graph with probability $2/3$ requires $\widetilde{\Omega}(n^{1/3})$ rounds of $\hybrid$.
\end{theorem}

For the proofs of \Cref{theorem:radius-bcc} and \Cref{theorem:radius-hybrid} we refer to \Cref{theorem:directed_girth-bcc} and \Cref{theorem:directed_girth-hybrid} respectively, where we employ similar arguments. In particular, we should remark that the simulation argument articulated by Kuhn and Schneider \cite[Lemma 7.3]{10.1145/3382734.3405719} can be directly extended for the graph $G_{k, \ell}^{x,y}$.

\subsubsection{Weighted Graphs}

Next, we further modify our construction in order to obtain stronger lower bounds for weighted graphs. Specifically, if $W$ represents the maximum-weight edge, we endow every edge of the graph $G_{k, \ell}^{x, y}$ with weight $W$, with the following exceptions: (i) all the edges belonging in paths connecting $V$ to $V'$; (ii) all the edges belonging in paths connecting $F$ to $T'$ and $T$ to $F'$; and (iii) all the edges belonging in the path $z_1 \rightarrow z_2 \rightarrow \dots \rightarrow z_{\ell + 1}$. We will represent the induced weighted graph as $G_{k, \ell, W}^{x, y}$.

\begin{claim}
For every node $u$ in $G_{k, \ell, W}^{x,y}$ besides the nodes in $U$ it follows that $\ecc(u) \geq \ell + 3W$.
\end{claim}

\begin{claim}
The radius $R$ of $G_{k, \ell, W}^{x, y}$ is $\ell + 2W$ if $x \cap y \neq \emptyset$; otherwise, $R =\ell + 3W$.
\end{claim}

Observe that for sufficiently large $W$ this claim implies an asymptotically $3/2$-gap depending on whether $x \cap y = \emptyset$. As a result, we are ready to establish the following theorem:

\begin{theorem}
    For any $\epsilon \in (0, 1/2]$, determining a $(3/2 - \epsilon)$-approximation for the radius of a weighted graph with probability $2/3$ requires $\widetilde{\Omega}(n^{1/3})$ rounds of $\hybrid$, assuming that $W = \omega(n^{1/3})$
\end{theorem}

\subsection{Diameter}

For this subsection we will employ some machinery developed in \cite{censorhillel2020distance} for solving in parallel multiple single-source shortest paths (SSSP) problems in $\hybrid$. Specifically, the following theorem applies when the set of sources is selected arbitrarily.

\begin{theorem}[\cite{censorhillel2020distance}, Theorem 1.5]
    \label{theorem:SSSP}
    Consider an $n$-node weighted graph. For any set of sources $U$ with $|U| = \mathcal{O}(n^{1/3})$, there exists a distributed algorithm in $\hybrid$ so that every node in the graph determines its exact distance from every source $s \in U$ in $\widetilde{\mathcal{O}}(n^{1/3})$ rounds with high probability.
\end{theorem}

Moreover, in the $n^x$-random-sources shortest paths (RSSP) problem we are given a set of sources sampled independently with probability $n^{x-1}$, for some $x \in (0,1)$,\footnote{A standard Chernoff bound argument implies that the number of sources is $n^{x}$ with high probability.} and the goal is to ensure that every node knows its distance from (all) the sampled sources. Interestingly, when the sources are selected at random we can solve substantially more SSSP problems in the same number of rounds:

\begin{theorem}[\cite{censorhillel2020distance}, Theorem 1.3]
    \label{theorem:RSSP}
    Consider an $n$-node weighted graph. If a set of nodes $S$ is sampled independently with probability $n^{x-1}$, there is a distributed algorithm in $\hybrid$ which guarantees that every node $v \in V$ knows its exact distance from every node in $S$ with $\widetilde{\mathcal{O}}(n^{1/3} + n^{2x-1})$ rounds with high probability.
\end{theorem}

In particular, the $n^{2/3}$-RSSP problem admits a solution in $\widetilde{\mathcal{O}}(n^{1/3})$ rounds of $\hybrid$, which is also tight (see \cite{censorhillel2020distance} for the details). Now consider a weighted graph such that $\Delta = \mathcal{O}(\polylog n)$. In the sequel, we let $N_L(u)$ represent the $L$-nearest vertices to $u \in V$ for some integer $L$. We will show how to implement the sequential algorithm of Roditty and Vassilevska W. \cite{DBLP:conf/stoc/RodittyW13} in $\hybrid$; we commence by reviewing their algorithm, which consists of the following steps:

\begin{enumerate}
    \item Select a random sample $S$ of vertices such that $|S| = \widetilde{\Theta}(n/L)$;
    \item Solve the SSSP problem for all $s \in S$, and determine the node $w \in V$ which maximizes the distance from the set $S$;
    \item Determine the set $N_L(w)$ and solve the SSSP problem for all $s \in N_L(w)$;
    \item Return as the estimate $\widetilde{D} := \max_{u \in V, s \in S \cup N_L(w)} \{ d(s, u) \}$.
\end{enumerate}

This algorithm is guaranteed to return a value $\widetilde{D}$ such that $\lceil 2/3 \cdot D \rceil - w(\cdot, \cdot) \leq \widetilde{D} \leq D$, for some edge weight $w(\cdot, \cdot)$. Importantly, it was subsequently observed in \cite{DBLP:conf/soda/ChechikLRSTW14} that this additive term can be eliminated via a simple modification; the idea is to expand $N_L(w)$ by a \emph{single level}, and then solve all the SSSP problems for the induced set $N_L'(w)$. Having assumed that $\deg(u) = \mathcal{O}(\polylog n)$ it follows that $|N_L'(w)| = \widetilde{\mathcal{O}}(L)$, and this modification does not alter the asymptotic running time up to polylogarithmic factors.

Now let us return to the implementation of this algorithm in the $\hybrid$ model. First, we let $L = n^{1/3}$; this choice will optimize the round complexity. The second step of the algorithm can be solved in $\widetilde{\mathcal{O}}(n^{1/3})$ rounds via the distributed algorithm of \Cref{theorem:RSSP}; note that the maximum of $d(u, S)$ over all $u$ can be determined in $\mathcal{O}(\log n)$ rounds via the aggregate-and-broadcast protocol in $\ncc$ (\Cref{lemma:AB}). Next, given that $|N_L'(w)| = \widetilde{\mathcal{O}}(n^{1/3})$, node $w$ can determine the set $N_L(w)$ via the local network in $\widetilde{\mathcal{O}}(n^{1/3})$ rounds; then, $w$ can broadcast the IDs of the nodes in $N_L'(w)$ via \Cref{lemma:TD} in $\widetilde{\mathcal{O}}(n^{1/3})$ rounds. As a result, every node in the graph knows the set $N_L'(w)$, and the third step in the algorithm can be implemented via \Cref{theorem:SSSP} in $\widetilde{\mathcal{O}}(n^{1/3})$, as the total number of sources is $\widetilde{\mathcal{O}}(n^{1/3})$. Finally, the last step can be easily implemented in $\mathcal{O}(\log n)$ rounds via \Cref{lemma:AB}. Consequently, we have established the following:

\begin{proposition}
    \label{proposition:diameter}
    For any weighted graph $G$ with $\Delta = \mathcal{O}(\polylog n)$ we can determine a $3/2$-approximation of the diameter in $\widetilde{\mathcal{O}}(n^{1/3})$ rounds of $\hybrid$.
\end{proposition}

For general graphs we can still provide the same guarantee if it happens that $|N_L'(w)| = \widetilde{\mathcal{O}}(n^{1/3})$; otherwise, one could execute the algorithm of Roditty and Vassilevska W. without performing the expansion on $N_L(w)$, but it is unclear whether this is better (in the worst case) than the naive approach. Nonetheless, we showed that there is essentially no reason \emph{not} to execute this algorithm in the $\hybrid$ model.

\section{Concluding Remarks}

In this work we have provided several new insights on the power of the $\hybrid$ model in distributed algorithms. Specifically, we first showed a deterministic protocol which ensures that every node learns the entire topology of a sparse graph in $\widetilde{\mathcal{O}}(\sqrt{n})$ rounds; some applications of this result for general graphs were presented via sparsification techniques, most notably leading to deterministic algorithms which come close to the best-known randomized algorithms for the fundamental all-pairs shortest paths problem. We also made a connection with the concept of low-congestion shortcuts, leading to a polylogarithmic-round algorithm for approximate Min-Cut, even if the local network is modeled via $\congest$. Finally, we established an $\widetilde{\Omega}(n^{1/3})$ round-complexity lower bound for computing the radius of a graph, implying that there is essentially no separation between the complexity of computing the radius and the diameter in $\hybrid$---at least for unweighted graphs. In conclusion, several interesting open questions have emerged given that we do not have matching upper and lower bounds for many of the studied problems. 
\bibliography{paper}

\appendix

\section{Detecting and Counting Subgraphs}
\label{appendix:detecting-counting}

We commence this section by giving several simple algorithms in $\hybrid$ for detecting and counting subgraphs of small diameter. In the sequel we assume for concreteness that we are searching for $r$-cycles, for some parameter $r \in \mathbb{N}$, but our results directly apply for subgraphs with diameter upper-bounded by $r$. As a warm-up, we establish the following:

\begin{proposition}
    \label{proposition:detecting_cycles-hybrid}
For any graph $G$ we can detect the existence of an $r$-cycle in $\lfloor (r-1)/2 \rfloor + \mathcal{O}(\log n/\log \log n)$ rounds of $\hybrid$.
\end{proposition}

\begin{proof}
First, every node in the graph performs \emph{flooding} for $\lfloor (r-1)/2 \rfloor$ rounds via the local network. Then, every node can determine (locally) whether it participates in a $r$-cycle; if it does, it can disseminate the information to the entire network via \emph{pointer jumping} in $\mathcal{O}(\log n / \log \log n)$ rounds via the global network; observe that the capacity of the global network might be exceeded, leading to the loss of messages, but this does not affect the performance, or indeed the correctness of the algorithm.
\end{proof}

Notice that the $\log \log n$ factor stems from the fact that each node can transmit to $\log n$ other nodes via the global network, slightly accelerating the pointer jumping process. Moreover, given that $\Omega(\log n/ \log \log n)$ rounds are required in $\ncc$ to broadcast a message \cite{10.1145/3323165.3323195}, and for certain graphs (e.g. a path with an $r$-cycle at the one end) the local network does not offer an asymptotic improvement to the broadcasting phase, the derived round complexity is optimal when $r$ is sufficiently small.

\begin{proposition}
    \label{proposition:counting_cycles-hybrid}
For any graph $G$ we can count the number of $r$-cycles in $ (r-1) \mathcal{O}(\log n)$ rounds of the $\hybrid$ model.
\end{proposition}

\begin{proof}
As in the previous protocol, every node performs flooding for $\lfloor (r-1)/2 \rfloor$ rounds via the local network. In this way, every node will be able to determine the number of $r$-cycles it participates in; we assume that every node will only take into account the cycles for which it has the smallest ID in order to avoid "double counting" during the aggregation process. Now observe that this number is at most $\binom{n}{r-1} \leq n^{r-1}$, and hence, it can be represented with $(r-1) \log n$ bits. Thus, we let every node split its number into $r-1$ messages, each representing a corresponding $\log n$-chunk of the binary representation. Then, we employ for $r-1$ iterations the aggregate-and-broadcast protocol of \Cref{lemma:AB} for the distributive aggregate function $\textsc{Sum}$, each time for a different chunk. This would require $(r-1) \mathcal{O}(\log n)$ rounds, and afterwards every node can perform the addition of the chunks locally.
\end{proof}

We should note that even these simple protocols improve exponentially over the best-known algorithms in the powerful $\congclique$ model which are based on matrix multiplication \cite{DBLP:journals/dc/Censor-HillelKK19}, leading to a round-complexity of $\mathcal{O}(n^{1 - 2/\omega})$, where $\omega < 2.32728596$ is the exponent of matrix multiplication \cite{alman2020refined}. Moreover, $\hybrid$ offers a substantial improvement over using only $\local$---which trivially requires $\Omega(D)$---or only $\ncc$. Indeed, in the sequel we present some lower bounds in the $\bcc$ model, which is substantially more powerful than $\ncc$  (potentially up to a logarithmic factor).

\subsection{Lower Bounds in \texorpdfstring{$\bcc$}{}}

In the $\bcc$ model we can establish $\widetilde{\Omega}(n)$ lower bounds for counting and detecting cycles. In particular, we commence with the following theorem:

\begin{theorem}
    \label{theorem:bcc-counting}
    Counting the number of $r$-cycles for $r = \mathcal{O}(1)$ with probability at least $2/3$ requires $\widetilde{\Omega}(n)$ rounds of $\bcc$.
\end{theorem}

Naturally, the same limitation applies for the $\ncc$ model comprising $\hybrid$. We have not seen \Cref{theorem:bcc-counting} being stated before in the literature, so we include a proof.

\begin{proof}[Proof of \Cref{theorem:bcc-counting}]
Consider an instance of the set disjointness problem with some arbitrary strings $x, y \in \{0,1\}^k$ for Alice and Bob respectively. We will reduce this problem to counting the number of $4$-cycles in a suitably constructed graph. To this end, consider a set of nodes $V = \{1, 2, \dots, r\}$, with $r$ being the smallest number such that $r(r-1)/2 \geq k$, and let $\phi$ be some injective function which maps every index $b \in [k]$ to an edge $\{i, j\}$, for some $i, j \in [r]$ with $i \neq j$. Let $E$ be the set of edges so that $\{i, j\} \in E$ iff $\{i, j\}$ belongs to the range of $\phi$ and $\phi^{-1}(\{i, j\}) = b$ with $x_b = 1$; in words, Alice's input string $x$ is encoded as a set of edges $E$. Similarly, Bob encodes his input string $y$ as a set of edges $E'$ on a set of nodes $V' = \{1', 2', \dots, r'\}$; we assume that Bob's encoding is performed via an injective mapping $\phi': b \mapsto \{i', j'\}$, where $\{i, j\} = \phi(b)$. Finally, we connect with edges the corresponding nodes of $V$ and $V'$, leading to a graph $G^{x,y}_k = (V \cup V', E \cup E' \cup \{ \{i, i'\} : i \in [r]\})$. This construction is illustrated in \Cref{fig:subb1}.

Now consider some distributed algorithm $\mathcal{A}$ which counts the number of $4$-cycles in the $\bcc$ model with probability $2/3$. We claim that Alice and Bob can employ $\mathcal{A}$ on the induced graph $G_k^{x,y}$ in order to solve the set disjointness problem. Specifically, first observe that the players can directly simulate the communication protocol of $\mathcal{A}$ in order to determine the total number of $4$-cycles in $G_{k}^{x,y}$. Moreover, we claim that a $4$-cycle in $G_k^{x,y}$ can (i) consist exclusively of nodes from $V$, (ii) consist exclusively of nodes from $V'$, or (iii) can be expressed as $i \rightarrow i' \rightarrow j' \rightarrow j \rightarrow i$. Importantly, Alice and Bob can compute the number of $4$-cycles of type (i) and (ii) respectively locally, without requiring any communication; thus, they can also learn the number of $4$-cycles of type (ii) and (i) respectively by sending $\mathcal{O}(\log n)$ bits with the other party---which has already determined locally the number of $4$-cycles induced by its corresponding nodes. Therefore, we have concluded that Alice and Bob have determined the number of $4$-cycles of type (iii). However, it follows that $x \cap y = \emptyset$ if and only if the number of such cycles is $0$. Therefore, they have managed to solve the set disjointness problem with probability $2/3$. As a result, we know from \Cref{theorem:disj} that $\Omega(k)$ bits were exchanged between nodes of $V$ and $V'$ during the execution of algorithm $\mathcal{A}$. Given that every node can only transmit $\mathcal{O}(\log n)$ (distinct) bits in the $\bcc$, we derive the desired $\Omega(k/(r \log n)) = \widetilde{\Omega}(n)$ round-complexity lower bound.
\end{proof}

A natural question is whether we can directly modify the technique we applied for counting $4$-cycles in order to establish an $\widetilde{\Omega}(n)$ lower bound for \emph{detecting} a $4$-cycle; observe that this would require that no $4$-cycles are present within the induced graphs on $V$ and $V'$. Unfortunately, there is an inherent barrier which relates to the \emph{extremal function} $\ex(n, H)$---the maximum number of edges an $n$-node graph could have without containing a subgraph isomorphic to $H$. It turns out that $\ex(n, C_4) = \Theta(n^{3/2})$, and more broadly, for any bipartite graph $H$ it is known that $\ex(n, H) = o(n^2)$. As a result, we cannot obtain an $\widetilde{\Omega}(n)$ lower bound for detecting an \emph{even} cycle with the method we described. However, this is not the case for \emph{odd} cycles, as implied by the classic Erd\H{o}s-Stone theorem \cite{bams/1183510246,Bollobas1998Modern}.

\begin{theorem}
    Let $H$ be an arbitrary graph with $r = \chi(H) > 2$. Then, 
    \begin{equation}
        \ex(n, H) = \left( \frac{r-2}{r-1} + o(1) \right) \binom{n}{2} = \Theta(n^2).
    \end{equation}    
\end{theorem}
As a result, we can show the following:

\begin{theorem}[\cite{DBLP:conf/podc/DruckerKO13}]
    Detecting the existence of a $5$-cycle with probability at least $2/3$ requires $\widetilde{\Omega}(n)$ rounds of $\bcc$.
\end{theorem}

The same theorem applies for detecting any $(2r+1)$-cycle when $r = \mathcal{O}(1)$. It should be noted that these connections were also made and articulated in \cite{DBLP:conf/podc/DruckerKO13}, but we state them here for completeness.

\begin{figure}[!ht]
\centering
\begin{subfigure}{.5\textwidth}
  \centering
  \includegraphics[scale=0.4]{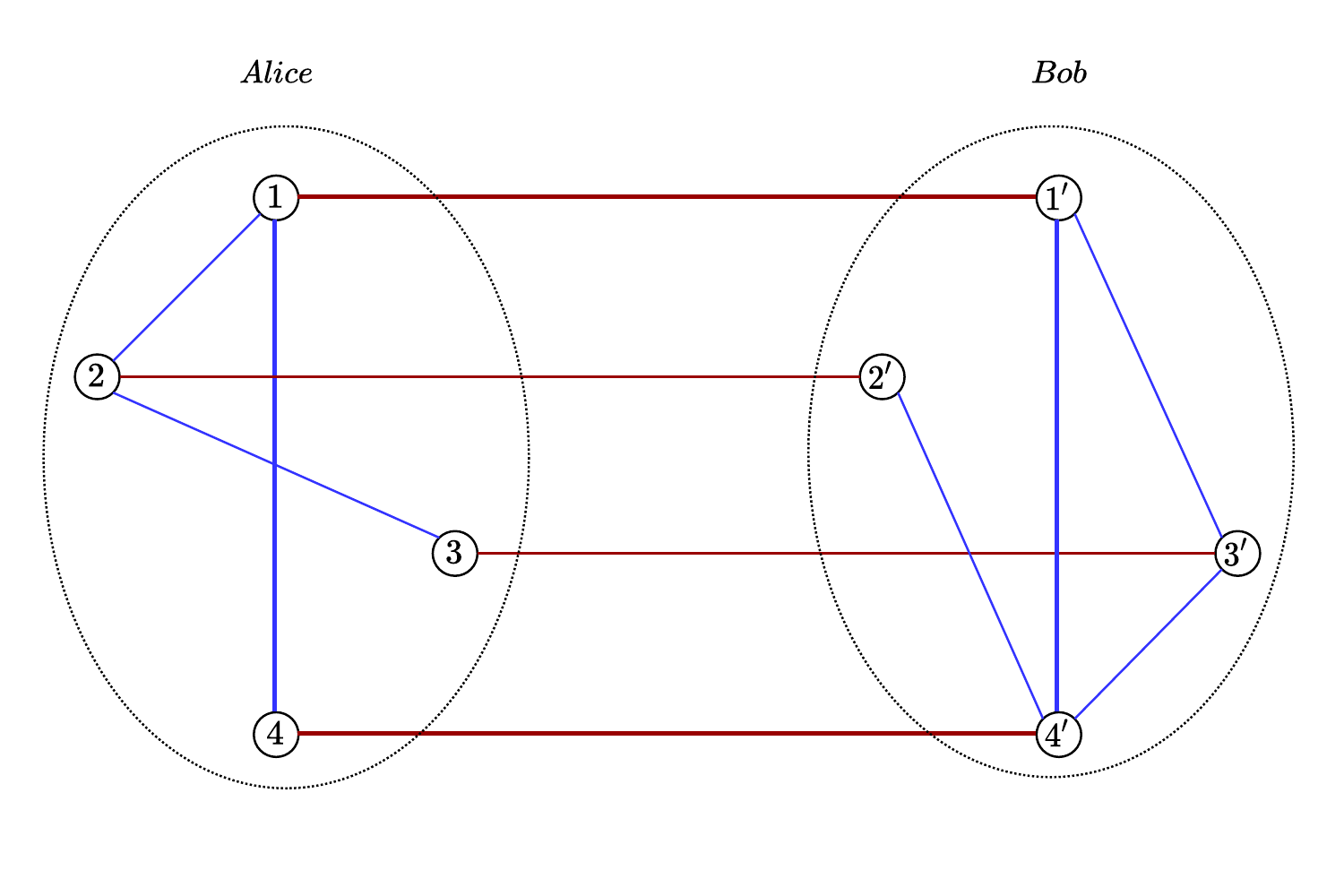}
  \caption{Reducing set disjointness to counting $4$-cycles.}
  \label{fig:subb1}
\end{subfigure}%
\begin{subfigure}{.5\textwidth}
  \centering
  \includegraphics[scale=0.35]{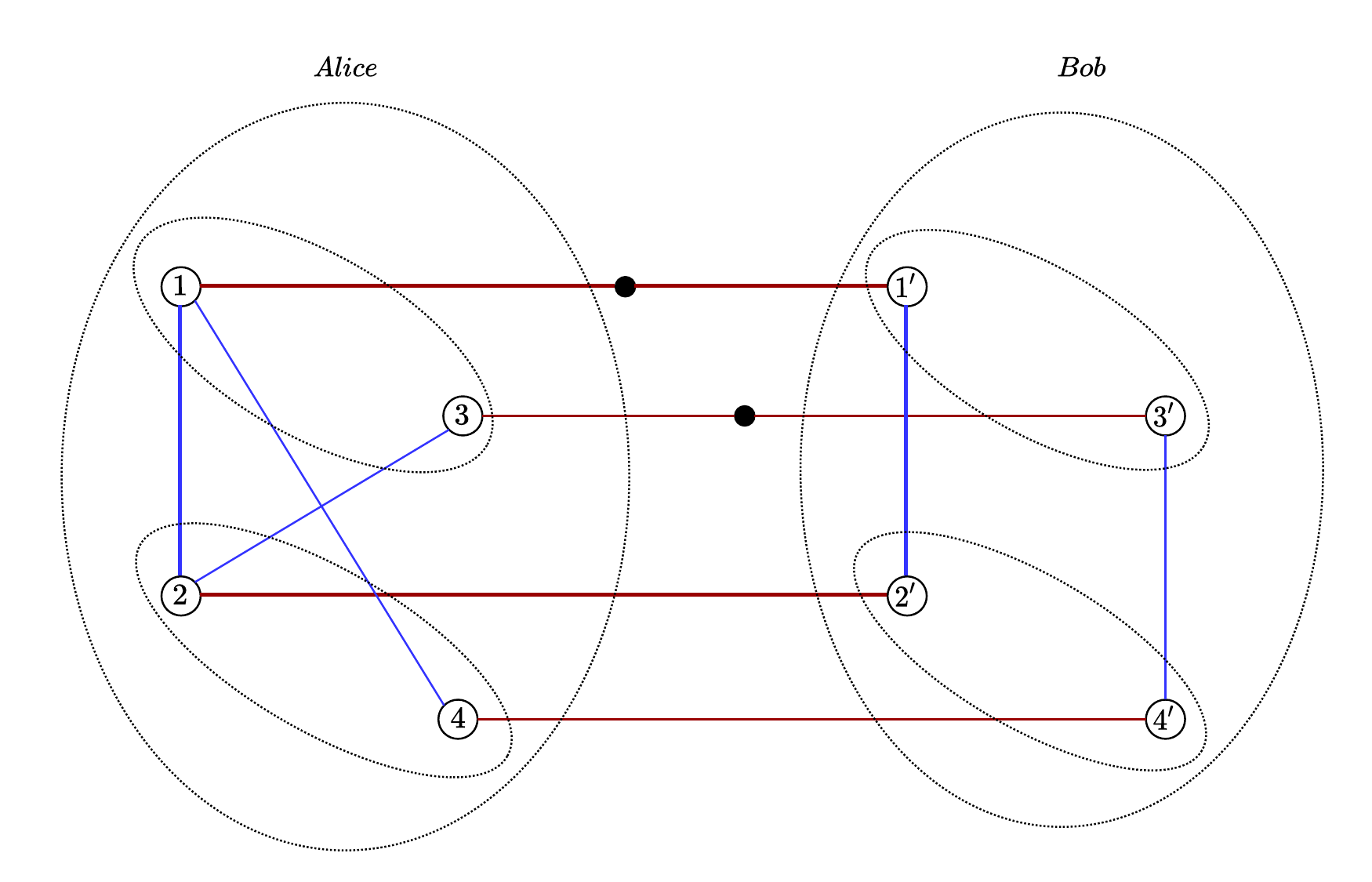}
  \caption{Reducing set disjointness to detecting $5$-cycles}
  \label{fig:subb2}
\end{subfigure}
\end{figure}

\section{Minimum Spanning Tree}
\label{appendix:MST}

Here we present an implementation of the classical algorithm of Boruvka for computing a minimum spanning tree in $\mathcal{O}(\log^2 n)$ rounds of $\congest + \ncc$. We stress that for this section the local edges are restricted to transfer $\mathcal{O}(\log n)$ bits per round, i.e. the local network is modeled with $\congest$. We also remark that the following implementation mainly uses primitives developed for the $\ncc$ in \cite{10.1145/3323165.3323195}, but our observation is that the local mode allows for a substantial speed-up in a key component of the algorithm, truncating the round complexity from $\mathcal{O}(\log^4 n)$ to $\mathcal{O}(\log^2 n)$ rounds.

In the MST problem the endpoints of every edge have to know at the end of the distributed algorithm whether the incident edge belongs to the MST. We let $V = \{u_1, \dots, u_n\}$ be the set of nodes of the corresponding graph, while we assume that the (edge) weights are unique, so that the MST is unique; note that this assumption is without any loss of generality given that ties can be broken based on the IDs of the incident edges. The standard Boruvka's algorithm with Heads/Tails clustering works as follows: At the beginning of the algorithm every node $u_i$ belongs to a separate connected component (or cluster) $C_i$. For every iteration, the nodes of every component $C_k$ determine the minimum-weight edge $\{u_i, u_j\}$ such that $u_i \in C_k$ and $u_j \in V \setminus C_k$. Next, every component flips a coin, and the minimum-weight edge $\{u_i, u_j\}$ is added to the MST only if the component of $u_i$ has flipped Heads and the component of $u_j$ has flipped Tails. Moreover, whenever an edge is added to the MST the corresponding components "merge". This idea is not part of the original Boruvka's algorithm, but it was instead introduced in \cite{DBLP:conf/soda/GhaffariH16,DBLP:conf/podc/GhaffariKS17}; observe that under Heads/Tails clustering all the merges are of "star" shape, which---among others---facilitates the design of fast merging protocols. The above process is repeated until only a \emph{single} connected component emerges. It is easy to see that this algorithm outputs with high probability the MST in $\mathcal{O}(\log n)$ iterations; indeed, although the Heads/Tails clustering reduces the merges occurring per iteration, it is innocuous as it only affects the round complexity by a constant factor (with high probability).

\subsection{Implementation in \texorpdfstring{$\congest + \ncc$}{}}

We will explain how to efficiently implement the previously described algorithm in the $\congest + \ncc$ model. First of all, every connected component $C_i$ will have a single leader node, denoted with $\ell(C_i)$, corresponding to some node within the component; naturally, the leader of every initial component $C_i = \{u_i\}$ will be $\ell(C_i) = u_i$. We stress that the leader will not explicitly know the nodes comprising its component, as this would require an overly amount of communication; instead, the invariance we will maintain is that every node knows its leader. In this context, our proposed implementation works as follows.

Consider a component $C_i$ at some iteration of the algorithm. We first need to ensure that the leader node knows the minimum-weight edge which connects $C_i$ to some other component. To this end, every node $u \in C_i$ communicates with its neighbors in order determine the subset of its neighborhood which lies on a different component; this can be performed in a single round of $\congest$ given that all the adjacent nodes can simply send the IDs of their corresponding leaders, and $u$ can compare these IDs to its own leader $\ell(C_i)$. Then, $u$ can compute locally the minimum-weight edge from $u$ to a different component. The next step is to determine the minimum over all the derived numbers within each component; this can be done in $\mathcal{O}(\log n)$ rounds of $\ncc$ via the aggregation algorithm of \cite{10.1145/3323165.3323195}. Consequently, the leader of the component will know the minimum-weight edge; during this process, it will also be useful to broadcast the ID of the component's leader which corresponds to the minimum-weight edge. 

Afterwards, the leader of every component flips a coin and observes either Heads or Tails. If the outcome is Heads, then the leader has to broadcast to every node in $C_i$ the ID of the leader in the corresponding component, which will serve as the new leader in the augmented component. This step can be implemented again in $\mathcal{O}(\log n)$ rounds via the \emph{multicast} algorithm in \cite{10.1145/3323165.3323195}. Otherwise, if the leader observes Tails it does not have to disseminate any information to the nodes within the component since the leader will remain the same. As a result, we have established the following:

\begin{theorem}
    \label{theorem:MST-hybrid}
    There exists a distributed algorithm which determines with high probability a minimum spanning tree in $\mathcal{O}(\log^2 n)$ rounds of $\congest + \ncc$.
\end{theorem}

\end{document}